%% file: msc-master.tex
\documentclass{paper}
\usepackage{breakcites}
\input{input/msc-notation}
\usepackage{mathrsfs}
\pgfplotsset{compat=1.18} 
\begin{document}

\include{input/msc-title}

\input{input/text/1-intro}
\input{input/text/2-model}
\input{input/text/3-comparability}

\input{input/text/4-first-fixation}

\input{input/text/5-literature}

\input{input/text/6-conclusion}
\appendix
\cleardoublepage%
\section{Proofs}
\input{input/text/app-2-proof-limit}
\input{input/text/app-3-proof-auxiliary}
\input{input/text/app-3-proof-T1}
\input{input/text/app-3-proof-comparability1}
\input{input/text/app-3-proof-T2}

\input{input/text/app-4-proof-edge}
\input{input/text/app-3-proof-luce}
\input{input/text/app-4-proof-T3}

%\input{input/text/app-4-proof-time-pressure}

%\input{input/text/app-5-proof-metro}
\section{Examples}
\input{input/text/app-2-ex-decoy}
\input{input/text/app-3-ex-T1}
\input{input/text/app-3-ex-T2}
\input{input/text/app-4-ex-T3}
\bibliographystyle{apacite}
\bibliography{msc-bib}
\end{document}

%% file: input/msc-notation.tex
\newcommand{\all}{\ensuremath{N}} %set of alternatives
\newcommand{\M}{\ensuremath{M}\xspace} % menu
\newcommand{\Mset}{\ensuremath{\mathcal{M}}\xspace} % set of subsets of a menu
\newcommand{\Nset}{\ensuremath{\mathcal{\all}}\xspace} % set of all menus
\newcommand{\set}[1]{\ensuremath{\{#1\}}\xspace}
\newcommand{\init}[2][]{
\ifthenelse{\isempty{#1}}{\ensuremath{\boldsymbol{\pi}^{#2}(\M)}}{\ensuremath{\pi^{#2}({#1},{\M})}}} % initial distirbution

\newcommand{\p}[3][]{\ensuremath{p^{#1}(#2,#3)}}% SCF
\newcommand{\bp}[1][]{\ensuremath{\boldsymbol{p}(#1)}}
\newcommand{\brho}[2][]{%
\ifthenelse{\equal{#1}{b}}{{\ensuremath{\boldsymbol{\rho}(\abs,\init[]{},\mat[]{\M})}}{\ensuremath{\boldsymbol{\rho}(\abs,\boldsymbol{\pi}^{#2}(\M),\mat[#2]{\M})}}%
}{\ifthenelse{\isempty{#2}}{\ensuremath{\boldsymbol{\rho}(\init[]{},\mat[]{\M)}}}{\ensuremath{\boldsymbol{\rho}(\init[]{},\mat[#2]{\M})}}%
}}
\newcommand{\ro}[4][]{
\ifthenelse{\isempty{#1}}{
\ensuremath{\rho({#3},\mat[#2]{#4})}}{\ensuremath{\rho({#3},\mat[#2]{\{#4\}})}}
} % SCF
\newcommand{\brhod}[2][]{%
\ifthenelse{\equal{#1}{d}}{\ensuremath{\boldsymbol{\rho}(\abs,\delta_{#2},\mat[]{\M})}}{\ensuremath{\boldsymbol{\rho}(\abs,\boldsymbol{\pi}^{#2}(\M),\mat[]{\M})}}}

\newcommand{\lra}[3][]{
\ifthenelse{\isempty{#1}}{
\ensuremath{\boldsymbol{\rho}(#2,#3)}}{\ifthenelse{\isempty{#2}}{\ensuremath{\rho(#1,#3)}}{\ensuremath{\rho(#1,#2,#3)}}}
} % SCF
\newcommand{\bra}[3][]{
\ifthenelse{\isempty{#1}}{
\ensuremath{\boldsymbol{\rho}(\abs,{#2},{#3})}}{\ensuremath{\rho(#1,\abs,{#2},{#3})}}
} % SCF

\newcommand{\mat}[2][]{\ifthenelse{\isempty{#1}}{\ensuremath{Q({#2})}}{\ensuremath{Q^{#1}({#2})}}} %tr prob matrix
\newcommand{\tmat}[1]{\ensuremath{\tilde{Q}({#1})}} % tilde tr prob matrix
\newcommand{\q}[3][]{
\ifthenelse{\isempty{#1}}{\ensuremath{q_{#2}(\M)}}{\ifthenelse{\equal{#1}{c}}{\ensuremath{q_{#2}(\{#3}\})}{\ensuremath{q_{#2}(#3)}}
}
} % tr prob; if empty M, if [c] curly brackets
\newcommand{\qp}[3][]{
\ifthenelse{\isempty{#1}}{\ensuremath{q'_{#2}(\M)}}{\ifthenelse{\equal{#1}{c}}{\ensuremath{q'_{#2}(\{#3}\})}{\ensuremath{q'_{#2}(#3)}}
}
} % tr prob; if empty M, if [c] curly brackets
\newcommand{\abs}{\ensuremath{\alpha}} %stopping prob
\newcommand{\period}{\ensuremath{\tau}} %time
\newcommand{\sa}[3][]{
\ifthenelse{\isempty{#1}}{\ensuremath{s_{#2}(\M)}}{\ifthenelse{\equal{#1}{c}}{\ensuremath{s_{#2}(\{#3}\})}{\ensuremath{s_{#2}(#3)}}
}}
\newcommand{\sav}[1][]{
\ifthenelse{\isempty{#1}}{\ensuremath{\boldsymbol{\gamma}(\M)}}{\ensuremath{\boldsymbol{\gamma}_{#1}(\M)}}
}
\newcommand{\dif}[3][]{
\ifthenelse{\isempty{#1}}{\ensuremath{d_{#2}(\M)}}{\ifthenelse{\equal{#1}{c}}{\ensuremath{d_{#2}(\{#3}\})}{\ensuremath{d_{#2}(#3)}}
}
} % %distortion from IA; if empty M, if [c] curly brackets
\newcommand{\cya}[1][]{{\ensuremath{C(#1)}}}

\newcommand{\cyaset}[1][]{\ensuremath{\mathcal{C}(#1)}}

\newcommand{\cyrset}[2][]{\ensuremath{\mathcal{C}_{#1}^{S}(\bp[#2])}}
\newcommand{\del}[2][]{\ensuremath{\delta_{#1}(\bp[#2])}}
\newcommand{\val}[1]{\ensuremath{u({#1})}}
\newcommand{\msc}[1][]{\ifthenelse{\equal{#1}{b}}{baseline MSC model}{\ifthenelse{\equal{#1}{l}}{limiting MSC model}{\ifthenelse{\equal{#1}{B}}{Baseline MSC model}{\ifthenelse{\equal{#1}{L}}{Limiting MSC model}{MSC model}}}}}
\newcommand{\G}{\ensuremath{G^{\M}}} %set of alternatives
\newcommand{\g}[1][]{
\ifthenelse{\isempty{#1}}{\ensuremath{g_{\M}}}{\ensuremath{g_{\M}((#1))}}}

\newcommand{\cycle}[3][]{
\ifthenelse{\isempty{#1}}{#3}{#2}} 
\newcommand{\btc}{bounded in a cycle\xspace}
\newcommand{\ai}{\ensuremath{i}\xspace}
\newcommand{\comp}[1][]{%
\ifthenelse{\equal{#1}{C}}{Comparability restriction\@\xspace}{comparability restriction\@\xspace}}
\makeatother

%% file: input/msc-title.tex
\begin{titlepage}
      \title{ \huge{Markov Stochastic Choice\thanks{I would like to thank my advisors Nick Netzer and Jakub Steiner for their invaluable guidance, support and constructive comments. I am grateful to Carlos Alós-Ferrer, Adam Dominiak, Sean Horan, Matthew Kovach, Marco LiCalzi, Fabio Maccheroni, Massimo Marinacci, Yusufcan Masatlioglu, Ariel Rubinstein, Gerelt Tserenjigmid, Kemal Yildiz, as well as seminar and conference participants at Bern, Bocconi, Virginia Tech, BRIC (2022), NASM (2021) and RUD (2020) for insightful discussions and feedback. I thank MIT and Drew Fudenberg for the hospitality during the academic year 2019/2020 and to the Swiss National Science Foundation (grant No.\ P1ZHP1-184169) for the generous financial support. This research has been mainly conducted while I was affiliated with the Department of Economics at the University of Zurich.
      }}}
      \author{Kremena Valkanova\thanks{E-mail: kremena.valkanova@gmail.com}}
      \date{September 16, 2024}
      \maketitle
\begin{abstract}
We examine the effect of item arrangement on choices using a novel decision-making model based on the Markovian exploration of choice sets. This model is inspired by experimental evidence suggesting that the decision-making process involves sequential search through rapid stochastic pairwise comparisons. Our findings show that decision-makers following a reversible process are unaffected by item rearrangements, and further demonstrate that this property can be inferred from their choice behavior. Additionally, we provide a characterization of the class of Markovian models in which the agent makes all possible pairwise comparisons with positive probability. The intersection of reversible models and those allowing all pairwise comparisons is observationally equivalent to the well-known Luce model. Finally, we characterize the class of Markovian models for which the initial fixation does not impact the final choice and show that choice data reveals the existence and composition of consideration sets.

\end{abstract}
\textit{Keywords}: stochastic choice, search, bounded rationality, attention.
\end{titlepage}

%% file: input/text/1-intro.tex
\section{Introduction}
We frequently encounter decision-making problems that involve locating, evaluating, and selecting the best option from a set of alternatives. Such decisions arise in various contexts, including choosing a product from a supermarket shelf or a vending machine, selecting a movie on a webpage, or deciding which article to read in a newspaper. Prior research demonstrates that the arrangement and adjacency of alternatives significantly impacts the decision-making process and  final choices \citep{Dreze1994, Chen2006, Keel2015,Huang2021, Mattis2024}. This mechanism of influencing decisions coexists with well-established attention-seeking techniques aimed at increasing the selection probability of target items\footnote{See \citet{Orquin2013} for a review.}, as well as with manipulations of the choice set composition, exploiting menu-dependent preferences.\footnote{\citet{Adler2024} provide a recent review over this vast literature in the marketing context.}

Predicting the effects of item arrangement on choices is crucial for choice architects, even though their objectives may vary. For instance, academics and consumer researchers, focused on eliciting unbiased preferences, would seek to minimize external influences that could distort decision-making. Alternatively, benevolent social planners may design choice environments to nudge individuals towards healthier options or to reduce cognitive load by simplifying the choice problem without influencing the decision-making process. Finally, sellers and platform designers might exploit product positioning to increase the selection frequency of certain items, thereby maximizing profits.

Without a thorough understanding of the underlying decision process, choice problem designers might disrupt consumer choices by nudging decision-makers in unforeseen or unintended directions. Moreover, identifying which types of decision-makers are more susceptible to these positioning effects is crucial for developing more effective and personalized interventions. Finding that certain decision-makers are unaffected by these factors justifies the use of varied presentation formats to simplify and guide the choice process -- for example, through list-based presentations or the integration of recommender systems on online platforms.

To address these challenges, we introduce and analyze a choice model we call Markov Stochastic Choice (MSC). This model allows us to examine the effects of item positioning on decision-making and captures patterns consistently observed across multiple disciplines, such as economics, marketing, psychology, and visual neuroscience. These patterns include rapid pairwise comparisons between alternatives, a tendency to compare nearby options, the final fixation often aligning with the selected item, and the inherently stochastic nature of decision-making.\footnote{Notable references include \citet{Russo1975}, \citet{Henderson1999}, \citet{Rayner1998}, \citet{Chandon2009}, \citet{Reutskaja2011}, \citet{Wedel2007}, \citet{Pieters2008}, \citet{Noguchi2014}, \citet{Shimojo2003}, and \citet{Armel2008}.}

The MSC decision-making process begins with an agent randomly selecting an alternative to view according to their initial beliefs and knowledge. In each period, the agent evaluates the current alternative against a competitor, transitioning to the competitor according to a probability distribution specific to the current alternative and independent of past transitions. This approach minimizes cognitive load by not relying on complete recollection of the past exploration. Transition probabilities are assumed to be consistent across different choice sets, meaning the relative likelihood of transitioning between any two alternatives remains unaffected by other items in the menu. The decision-making process may terminate in each period with a given stopping probability, after which the agent selects the most recently considered alternative. The described procedure corresponds to a discrete-time Markov chain, where each state represents an alternative in the choice set and the transitions are captured by eye movements, or saccades. With diminishing stopping probability, \ie when there is no time pressure, we show that the choice function of the decision-maker converges towards the limiting distribution of the Markov chain.

Various factors can influence stochastic pairwise transitions. Eye fixations towards an item might be solely driven by its salience and attention-grabbing potential regardless of its value, as suggested by \citet{Reutskaja2011} and \citet{Milosavljevic2012}. Alternatively, probabilistic transitions between pairs of alternatives might reflect agent's noisy preferences. The literature on stochastic choice offers several explanations for this noise, including random shocks to preferences, evidence accumulation, bounded rationality, and random mistakes (see \citet{Agranov2017}). 
%Stochastic transitions between alternatives also facilitate comprehensive exploration of the choice set, preventing the decision maker from becoming confined to alternatives with high local but not global utility. 
Our model offers a simple framework to employ widely used models of binary stochastic choice such as the Drift-Diffusion Model \citep{Ratcliff1978} to study multi-alternative choice. 

We first consider nudging interventions involving restrictions of certain pairwise comparisons, achieved by either placing two items too far apart or not presenting them as related items in a recommender system, while still keeping the set of reachable items the same. We find that the types of decision makers who are not influenced by such presentations of the choice problem are characterized by following a reversible Markov process, \ie the average number of saccades from one alternative to another should be the same as in the reverse direction in the long run. Since the choices are robust to any comparability restriction, we could infer the reversibility property if we find that no rearrangement of the choice set leads to a change in the stochastic choices. What is more, the choice function itself reveals the reversibility property of the rationalizing model. This is the case when there is no cycle of alternatives for which the relative choice probabilities between each of the subsequent pairs is higher in the current choice set than in the corresponding binary sets.

Next, we identify decision makers with greater cognitive capabilities who can make all pairwise comparisons directly regardless of arrangement, \ie their choices are rationalizable with an \msc[] we call fully comparable. The corresponding choice function is such that any pair, for which the relative choice probability in the current choice set differs from that in binary sets, belongs to a cycle where all differences in relative choice probabilities between the current and binary menus have the same sign. We term the choice function between such pairs as ``bounded in a cycle''. A violation of this property would imply that the decision-maker was not able to compare the pair directly.

Interestingly, when considering the intersection of reversible and fully comparable MSC models, we find that they are characterized by two well-known axioms: positivity and Independence of Irrelevant Alternatives (IIA). Thus, this special class of MSC models is observationally equivalent to the well-known Luce model \citep{Luce1959}. Additionally, we demonstrate that the relative transition probabilities between each pair are determined solely by the relative utility of the items and are not influenced by other factors such as salience.

Finally, we analyze the potential to influence decision-making by altering the initial fixation probability. In practical applications, increasing the likelihood of starting the decision-making process at a specific target alternative can be achieved through various techniques, such as central placement, eye-catching packaging, default or sponsored item displays, and effective advertising \citep{Chandon2009,Orquin2013}. Our model predicts that under time pressure, the initial fixation can significantly affect final choices. When agents have unlimited time, such nudging techniques are only effective if distinct consideration sets exist, \ie the Markov process is reducible. We demonstrate that for a stochastic choice function to be rationalizable with an irreducible model, a cycle must exist involving the entire menu where all subsequent pairs are bounded in a cycle. Additionally, the stochastic choice function reveals the precise communicating classes, making these interventions predictable and effective.

The remainder of the paper is structured as follows: Section~\ref{sec:model} defines the choice model. Section~\ref{sec:comp-restr} analyzes the impact of comparability restrictions and identifies reversible and pairwise comparable models. Section~\ref{sec:fixation} examines the effect of initial fixation probability on final choices and characterizes the irreducible MSC model. Section~\ref{sec:literature} reviews the relevant literature. Section~\ref{sec:conclusion} concludes.

%% file: input/text/2-model.tex
\section{Model}
\label{sec:model}
Let $\all$ be a finite set of all alternatives. A menu (or choice set) \M is a non-empty subset of $\all$. The set of all menus is denoted by \Nset and the set of all non-empty subsets of a menu \M is denoted by \Mset. 
\subsection{Stochastic choice functions}
A stochastic choice function is a mapping $p:\all\times \Nset \rightarrow [0,1]$ such that $\sum_{i\in \M}\p{i}{\M}=1$ and $i\not\in \M$ implies $\p{i}{\M}=0$. The function $\p{i}{\M}$ is interpreted as the probability of choosing alternative \ai from a menu \M. We denote a row vector of choice probabilities from a menu \M by \bp[\M]. The stochastic choice function on all menus is denoted by $\bp[\Nset]$. A stochastic choice function is positive if $\p{i}{\M}>0$ for all $i\in\M$ and all $\M\in\Nset$.

Let $\del[ij]{\M}=\p{i}{\M}\p[]{j}{\set{i,j}} -\p[]{i}{\set{i,j}}\p{j}{\M}$ be a weighted difference in choice probabilities from menu~$\M$ between two alternatives~$i$ and~$j$. A cycle of alternatives $\cya[\M]$ on $\M=\set{i_1,i_2,\dots, i_{|\M|}}$ is a set of ordered pairs such that $\cya[\M] = \set{(i_1,i_2),(i_2,i_3),\dots,(i_{|\M|},i_1)}$. The set of all cycles on~$\M$ is denoted by $\cyaset[\M]$. 
%We define the following ratio of binary choice probabilities for a given \cy as
%\[r(\cy)= \prod_{(k,l)\in\cy}\frac{\p[c]{l}{k,l}}{\p[c]{k}{k,l}}.\]

\begin{definition}
A cycle of alternatives $\cya[\M']$ is called sign-consistent on $\M'\in\Mset$ w.r.t.\ \bp[\M] if $\text{sgn}(\del[ij]{\M})=\text{sgn}(\del[kl]{\M})$ for all $(i,j),(k,l)\in\cya[\M']$.
\end{definition}
The set of all sign-consistent cycles on all subsets of $\M$ w.r.t.\ \bp[\M] is given by $\cyrset[]{\M}$. To indicate the sign of $\del[ij]{\M}$, we use the notation $\cyrset[0]{\M},\cyrset[+]{\M}$, $\cyrset[-]{\M}$, respectively.
%\begin{definition}
%A sign-consistent cycle of alternatives on $\M'\in\Mset$ w.r.t.\ \bp[n]{\M} is given by \cyr[]{\M'} such that $\cyr[]{\M'}\in\cyaset[]{\M'}$ and $$\prod_{(k,l)\in\cyr[]{\M'}}\p[c]{k}{k,l}\leq \prod_{(k,l)\in\cyr[]{\M'}}\p[c]{l}{k,l}$$ and $\p{j}{\M}\p[c]{k}{j,k}\neq \p[c]{j}{j,k}\p{k}{\M}$ for all $(j,k)\in\cyr[]{\M'}$.
%\end{definition}
\subsection{General choice procedure}
A decision maker evaluates the alternatives in a given menu~$M$ as follows. The agent has an initial fixation probability distribution $\init[]{}$ over the alternatives in the menu, determining the probability of starting the exploration process at each alternative. At time period $\period=0$, the agent draws an alternative~$i \in \M$, which becomes the current best alternative.

In the next period, alternative~\ai is compared to another available alternative~$j$. If~$j$ wins the comparison, the decision maker transitions from~$i$ to~$j$, making~$j$ the best alternative at period $\period=1$. These transitions are stochastic, with the probability of transitioning from~$i$ to~$j$ denoted by~$\q{ij}{}$. A transition probability matrix \mat[]{\M} contains all pairwise transition probabilities. 

The transition probabilities may be influenced by decision makers' preferences and cognitive constraints, as well as the presentation of the choice problem. These factors, in turn, are affected by items' salience, positioning, and default option status. Consequently, we assume that the transition probabilities are time-independent, \ie revisiting alternatives does not affect the probability to make subsequent transitions. However, we allow transition probabilities between each pair of alternatives to be menu-dependent, \ie $\q{ij}{} \neq \q[r]{ij}{\M'}$.
We impose the following three assumptions on the matrices \mat[]{\M} for all $\M\in\Nset$.
\begin{assumption}[Prolonged consideration]
For all $\M\in\Nset$ and $i,j\in \M$ holds
$\q[]{ii}{\M}=1-\sum_{j\neq i}\q{ij}{}>0$.
\label{a-consideration}
\end{assumption}
The assumption of prolonged consideration posits that there is a positive probability that no transition occurs in a given period, regardless of the current best alternative. Furthermore, it asserts that $\mat{\M}$ is a right stochastic matrix.

\begin{assumption}[Pairwise comparability on binary sets]
For all $i,j\in\all$, $\q[c]{ij}{i,j}=0$ implies that $\q[c]{ji}{i,j}>0$.
\label{a-pc}
\end{assumption}
Pairwise comparability on binary sets means that when choosing between two alternatives, the decision maker can transition in at least one direction, \ie the items are comparable. It implies that cognitive capacity should not limit transitions when evaluating only a pair of alternatives.

\begin{assumption}[Transition ratio independence of irrelevant alternatives (TR-IIA)]
For all $\M\in\Nset$ and $i,j\in \M$ holds
$\q[c]{ij}{i,j}\q{ji}{}=\q[c]{ji}{i,j}\q{ij}{}$.
\label{a-TR-IIA}
\end{assumption}
The TR-IIA assumption sets a minimal consistency requirement in decision-making across choice sets: the ratio of transition probabilities remains unaffected when the menu expands. This assumption arises when transition probability depends on the utility difference between two items and their distance -- whether physical (positioning on a shelf or screen) or subjective (similarity or ease of comparison).\footnote{See for example \citet{Russo1975}, \citet{Reutskaja2011} and \citet{Cerreia-Vioglio2018}.} Since these components remain unchanged by adding more items to the choice set, the transition ratio between the pair remains constant.
%The decision maker focuses on a competitor~$j$ with probability $\sa{ij}{}$. This parameter captures the ability of~$j$ to attract attention away from~\ai and can be interpreted as~\ai 's relative salience against~$j$. Then, the two alternatives are compared independently from the other alternatives in the menu. The competitor wins the comparison with probability $\tr{ij}$ and loses with probability~$1-\tr{ij}$, which we refer to as the acceptance or rejection probabilities, respectively. Therefore, the total probability to switch incumbents from \ai to $j$ is given by:
%\[\q{ij}{}=\sa{ij}{}\tr{ij}\geq 0.\]
%The probability that~\ai remains the incumbent\footnote{Note that the probability to remain at the same incumbent is assumed to be positive, because in some periods the incumbent might not be challenged by another alternative.} is given by ${\q{ii}{}=1-\sum_{j\neq i}\q{ij}{}>0.}$ We denote the transition probability matrix with \mat[]{\M}. 

The decision-making procedure is repeated analogously in consecutive periods. The process is terminated with probability~$\abs \in (0,1)$ in each period. When the process stops, the agent chooses the last considered alternative. 

The choice procedure is equivalent to a finite-state discrete-time Markov chain with each state representing one alternative in menu~$\M$, initial distribution~$\init[]{}$, transition probability matrix~$\mat[]{\M}$, and stopping probability~$\abs$. The stochastic choice function generated by the model for each menu takes the following form:
\begin{equation}
\bra[]{\init[]{}}{\mat{\M}} = \sum_{\period=0}^\infty \abs\init[]{}(1-\abs)^{\period} \mat[]{\M}^{\period},
\label{msc-b}
\end{equation}
where each term in the sum represents the probability that the process stops at a given period multiplied by the probability of each alternative to be the current best at that time period. 

\subsection{Choice procedure without time constraints}
Consider the case in which the probability that the decision process is terminated approaches zero, hence the pressure to make a decision diminishes. If the underlying Markov chain is irreducible and aperiodic\footnote{Note that the assumption of prolonged consideration made earlier ensures that the Markov chain of any \msc[] is aperiodic. Recall that a communicating class of a Markov chain over~$\M$ is defined as a set of states $\M'\in\Mset$, which communicate, \ie for each pair of states $i,j\in \M'$ there is a sequence of states $k_1,k_2,\dots k_h\in\M'$ such that $\q{ik_1}{}\q{k_1k_2}{}\dots\q{k_hj}{}>0$. A Markov chain is irreducible when it has a single communicating class, hence the process can reach any state irrespectively of the initial state. Since the Markov chain of a \msc[] is finite, irreducibility implies positive recurrence.}, hence ergodic, a limiting distribution exists and it is the unique stationary distribution. 
% $\rho(M) = \lim_{\period\rightarrow \infty} \rho_{\period}({\M})$. 
Naturally, if the stopping probability approaches zero, the generated stochastic choice function converges to the stationary distribution of the Markov chain and does not depend on the initial distribution.
\vspace{1pt}
\begin{proposition}
\label{prop:limit}
A stochastic choice function $\bra[]{\init[]{}}{\mat{\M}}$ generated by an ergodic Markov chain converges to its stationary distribution as $\alpha \rightarrow 0$.
%, hence $\lim\limits_{\abs \rightarrow 0}\bp[a]{\M} = \rho(\mat[]{\M})$.
%If the stopping probability approaches zero the stochastic choice function coincides with the stationary distribution of the Markov chain defined by  $\langle \mat[]{\M},\init[]{},\abs \rangle$.
%
\end{proposition}
\begin{proof}
%Fix a menu $\M\in\Mset$. We rearrange equation~\eqref{def} in the following way:
% \begin{equation*}
%\begin{gathered}
%\bp[n]{\M}(I-(1-\abs)\mat[]{\M})= \abs \init[]{} \\
%\bp[n]{\M}(I-\mat[]{\M}) = \abs (\init[]{} -\bp[n]{\M}\mat[]{\M}).
%\end{gathered}
%\end{equation*}
%Therefore,
%\[\lim_{\abs\rightarrow 0}\bp[n]{\M}(I-\mat[]{\M})=\boldsymbol{0}.\]
%Hence, $\bp[n]{\M}$ converges to the unique stationary distribution of the Markov chain with transition probability matrix $\mat[]{\M}$ as $\alpha \rightarrow 0$.
%*************************
% old proof 
See Appendix \ref{proof:limit}.
\end{proof}
We denote the generated stochastic choice function as $\alpha \rightarrow 0$ by $\lra[]{\init[]{}}{\mat{\M}}$. 
\begin{definition}
A stochastic choice function $\bp[\Nset]$ is rationalizable with an \msc[] if there exists a tuple $\langle \mat[]{\M},\init[]{}\rangle$ satisfying Assumptions 1 -- 3 such that $\bp[\M]=\lra[]{\init[]{}}{\mat{\M}}$ for all $\M\in \Nset$, where ${\lra[]{\init[]{}}{\mat{\M}}(I-\mat[]{\M}) = \boldsymbol{0}}$.
%\begin{itemize}
%\item[--] a \msc[b] if there exists a tuple $\langle \mat[]{\M},\init[]{},\abs \rangle$ satisfying Assumptions 1 -- 3 such that $\bp[\M]=\bra[]{\init[]{}}{\mat{\M}}$ for all $\M\in \Nset$ and
%\item[--] 
%\end{itemize}
\end{definition}

The above equality follows immediately from Proposition~\ref{prop:limit}. Note that if the Markov chain is irreducible, the choice function is unique. If it is not irreducible, each communicating class has its own unique stationary distribution and hence, the choice function depends on the initial distribution $\init{}$. %The following example shows the different rationalizability conditions for the baseline and limiting \msc[] on binary menus.
%\begin{example}A stochastic choice function rationalizable with a \msc[b] satisfies
%\begin{equation}
%\p[c]{i}{i,j}=\frac{\abs \pi(i,\{i,j\}) + (1-\abs) \q[c]{ji}{i,j}}{\abs+(1-\abs)(\q[c]{ij}{i,j}+\q[c]{ji}{i,j})}.
%\end{equation}
%On the other hand, if the stochastic choice function is rationalizable with a \msc[l] then 
%\begin{equation}
%p(i,\{i,j\})=\frac{\q[c]{ji}{i,j}}{\q[c]{ij}{i,j}+\q[c]{ji}{i,j}},
%\end{equation}
%and hence the ratio of transition probabilities between~\ai and $j$ is observable in the limiting model as it is equal to the ratio of the respective choice probabilities.
%\label{exp-binary}
%\end{example}

%% file: input/text/3-comparability.tex
\section{Comparability restrictions}
\label{sec:comp-restr}
In this section, we examine how decision makers' inability to make certain pairwise comparisons affects their final choices. Such comparability restrictions between pairs of alternatives may arise from cognitive limitations, especially when many alternatives are presented simultaneously, as in the supermarket example mentioned in the Introduction. This concept is supported by empirical evidence showing that jumps between fixations mostly occur between nearby alternatives \citep{Chandon2009}. Hence, different arrangements of the same alternatives can impact their pairwise comparability. Figure~\ref{fig:shelf} illustrates such a scenario, where the decision maker cannot directly compare diagonally positioned pairs when they are displayed on a grid.
	\input{input/figure/fig-shelf}

Additionally, a choice problem designer might restrict the visibility of the relevant items, limiting the decision maker's ability to compare certain pairs. For example, changes in the network of related products of a platform's recommendation system can influence the final purchasing decision, even if the decision maker would eventually consider the same set of items without time constraints.

We formalize the concept of \comp{}s below. We use the notation $\odot$ for element-wise matrix multiplication.

%\begin{definition}
%A \comp{} is a mapping $f: [0,1]^{|\M|\times|\M|}\rightarrow [0,1]^{|\M|\times|\M|}$ such that if $f(\mat{\M})=\mat[']{\M}$ then for at least one pair $i,j\in\M$ holds $\q{ij}{}\neq \qp{ij}{}=\qp{ji}{}=0$, while preserving the communicating classes. For all remaining pairs $k,l\in\M$ holds $\qp{kl}{}=a\q{kl}{}$ with $a>0$.
%\end{definition}
\begin{definition}
A \comp{} of an \msc[] $\langle \mat[]{\M},\init[]{}\rangle$ is a symmetric matrix $R(\M)$ with elements $r_{ij}(\M)\in\set{0,c}$ for $i\neq j$ and $c>0$ such that $R(\M)\odot \mat{\M}$ is a right stochastic matrix with the same communicating classes as $\mat{\M}$.
 %and $r_{ii}(\M)=\frac{1-\sum_{j\neq i}r_{ij}(\M)\q{ij}{}}{1-\sum_{j\neq i}\q{ij}{}}$
\end{definition}
In other words, \comp{}s prevent the decision maker from making a direct transition between pairs of alternatives, while keeping the reachable set of alternatives unchanged and scaling the transition probabilities between non-restricted pairs with the same factor. We say that an \msc[] is \textit{robust to a \comp{}}~$R(\M)$ if $\lra[]{\init[]{}}{\mat{\M}} = \lra[]{\init[]{}}{R(\M)\odot\mat{\M}}$.

In the following, we identify and characterize those classes of \msc[]s, which are robust to any \comp{}s. This is useful for two main reasons. First, it justifies the use of different presentation formats and simplifications of the choice problem (say, presenting items as a list instead on a grid) in cases when decision makers are not affected by \comp{}s. Second, \comp{}s can aid decision-makers in overcoming behavioral biases such as the attraction effect as discussed in the example below.
\begin{example}
\label{ex:decoy}
The attraction effect occurs when adding an asymmetrically dominated alternative to a choice set increases the choice probability of the dominating alternative. This phenomenon has been consistently documented in a variety of choice contexts and across different species \citep{Huber1982, Castillo2020}. The attraction effect is inconsistent with standard stochastic choice axioms such as regularity and independence of irrelevant alternatives. The \msc[] can accommodate this behavioral bias and provide suggestions for its mitigation. 

More specifically, in any \msc[] over $\all= \set{i, j, k}$ with strictly positive transition probabilities except $\q[c]{ik}{i,k}=q_{ik}(\all)=0$ (\ie, the target is alternative~$i$ and its decoy is alternative~$k$) the relative choice frequency of the target vs.\ the competitor increases when the decoy is present.
%target is chosen more often relative to the competitor when the decoy is present:
%$$\frac{\lra[i]{}{\mat{\all}}}{\lra[j]{}{\mat{\all}}}>\frac{\lra[i]{}{\mat{\set{i,j}}}}{\lra[j]{}{\mat{\set{i,j}}}}.$$
Moreover, the target is chosen more frequently from the triple than from the pair in absolute terms, $\lra[i]{}{\mat{\all}}>\lra[i]{}{\mat{\set{i,j}}}$, if and only if $q_{ki}(\all)>q_{ji}(\all)$, \ie the decision maker is more likely to transition to the target from the decoy than from the competitor. Thus, the attraction effect should vanish if the comparability between the target and the decoy is restricted, \ie $r_{ik}(\all)=r_{ki}(\all)=0$. More details and proofs can be found in Appendix~\ref{proof:decoy}.
\end{example}

\subsection{Reversible MSC models}
In this subsection we characterize the Markovian decision-making process of agents whose choices are robust to \comp{}s. This class of \msc[]s satisfies a well-known property of Markov chains called \textit{reversibility} \citep{Levin2017}. A Markov chain is reversible if it satisfies the following detailed balance conditions
\begin{equation}
\q{ji}{}\lra[j]{}{\mat{\M}}=\lra[i]{}{\mat{\M}}\q{ij}{}, \forall i,j \in\M.
\label{balance}
\end{equation}
The detailed balance equations postulate that the flow of probability mass is balanced for each pair of states. In our setting, detailed balance means that the average number of saccades from~$i$ to~$j$ should be the same as the saccades from~$j$ to~$i$ in the long run.\footnote{ A necessary and sufficient condition for reversibility of a Markov chain is the Kolmogorov's criterion. It states that a Markov chain is reversible if and only if for all cycles $\cya[\M']\in\cyaset[\M']$ and all $\M'\in\Mset$ holds
	\begin{equation}
	\label{kolmogorov}
	\prod_{(k,l)\in\cya[\M']}\q{kl}{}=\prod_{(k,l)\in\cya[\M']}\q{lk}{}.
	\end{equation}
	}
	 We call an \msc[] reversible on $\M$ if the corresponding Markov chain with states $\M$ satisfies reversibility. We say that an \msc[] is reversible if it is reversible on all $\M\in\Nset$. 

The following theorem is our first main result. It characterizes reversible \msc[]s in terms of \comp{}s and decision maker's choices from a given menu.

\begin{theorem}
\label{T1} The following are equivalent: A stochastic choice function $\bp[\M]$ is
	\begin{enumerate}
	\item[(i)] rationalizable only by reversible \msc[]s on~$\M,$
	\item[(ii)] robust to any \comp{}s of a rationalizing model,
	\item[(iii)] such that $\cyrset[]{\M} = \cyrset[0]{\M}$. 
%	
%	for any regular cycle $\cyr[]{\M'}$ there is some pair $(k,l)\in\cyr[]{\M'}$ for which
%	$\p{k}{\M}\p[c]{l}{k,l}<\p{l}{\M}\p[c]{k}{k,l}$.
%	
%	\[\frac{\p{k}{\M}}{\p{l}{\M}}\not\in  \left(\frac{\p[c]{k}{k,l}}{\p[c]{l}{k,l}}, r(\cy)\frac{\p[c]{k}{k,l}}{\p[c]{l}{k,l}}\right).\]
	\end{enumerate}	        
\end{theorem}
\begin{proof}
See Appendix~\ref{proof:T1}.
\end{proof}
Statements (i) and (ii) of the above theorem imply that if the decision-making process is reversible, restricting some pairwise comparisons does not affect final choices. Additionally, if the process is reversible without \comp{}s, different arrangements of alternatives, as shown in Figure~\ref{fig:shelf}, do not impact choice frequencies, even if the decision maker can only compare nearby options. Similarly, these decision makers' choices are unaffected by the number of related products on an e-commerce website, as long as they can access the same items from any starting point.

Following Theorem~\ref{T1}, an analyst can use \comp{}s to determine if the decision-making process is reversible on all~$\M\in\Nset$. To this end, the analyst needs to verify that the final choices do not change in response to \textit{any} of the possible \comp{}s. When the stochastic choice function is positive, then the analyst only needs to check the effect of \comp{}s when the menu equals the complete set of alternatives $\M=\all$. This is because if stochastic choice data on~$\all$ is rationalizable by a reversible \msc[], it can be shown using TR-IIA and Kolmogorov's condition that the \msc[] is reversible on all subsets. 

Statement (iii) offers an alternative way to identify reversible \msc[]s by examining the properties of the stochastic choices. It states that there is no sign-consistent cycle of alternatives for which the ratio of choice probabilities is either strictly higher or strictly lower in the larger choice set compared to the binary sets for each pair of alternatives in the cycle. The characterization result is particularly useful if an analyst wants to verify reversibility, but cannot control \comp{}s, such as when they are due to the decision maker's unobservable cognitive limitations.
%The following corollary shows the bounds to the choice probability ratios for positive choice functions resulting from this property.
%\begin{corollary}
%If a positive stochastic choice function satisfies $\cyrset[]{\M} = \cyrset[0]{\M}$, then for any cycle $\cyr[-]{\M'}{\M}$ and any $\M'\in\Mset$ there is some pair $(k,l)\in\cyr[-]{\M'}{\M}$ for which
%\[\frac{\p{k}{\M}}{\p{l}{\M}}\not\in  \left(\frac{\p[]{k}{\set{k,l}}}{\p[]{l}{\set{k,l}}}\prod_{(i,j)\in\cyr[-]{\M'}{\M}}\frac{\p[]{j}{\set{i,j}}}{\p[]{i}{\set{i,j}}},\frac{\p[]{k}{\set{k,l}}}{\p[]{l}{\set{k,l}}}\right).\]
%\end{corollary}

The two characterizations of reversible \msc[]s can complement each other to make predictions and uncover the true generating model's properties. The equivalence of Theorem~\ref{T1}(i) and (iii) implies that if only choice data for a particular menu is observed, \comp{}s can simplify the choice problem without affecting the decision maker's choices. Additionally, if a stochastic choice function violates Theorem~\ref{T1}(iii) and is rationalizable by a non-reversible and a reversible model, \comp{}s can help distinguish between the two. If any such restriction changes the final choices, the true generating model is non-reversible. The following example illustrates Theorem~\ref{T1}.

\input{input/table/tab-T1}
\input{input/figure/fig-T1}
\begin{example}
Consider the choice function from binary sets given in Table~\ref{tab:example} and let $\bp[\set{i,j,k,l}] = \bp[\all] = (0.2,0.2,0.4,0.2)$\footnote{Note that although this stochastic choice function and the ones used in Examples~\ref{ex:T1}-\ref{ex:T3} all satisfy regularity, it is not a necessary condition for rationalizability with an \msc[] as opposed to other models of stochastic choice such as random utility models.}. 
We first verify that the stochastic choice function satisfies Theorem~\ref{T1}(iii), \ie we need to show $\cyrset[-]{\all}=\cyrset[+]{\all}=0$ for any $\M'\in\Nset$. The only pairs of alternatives whose choice ratio from $\all$ do not equal their choice ratio from the binary set are $(i,k)$, $(j,k)$, and $(k,l)$ and their reverse pairs, \ie $\del[ik]{\M}\neq 0$, $\del[jk]{\M}\neq 0$, and $\del[kl]{\M}\neq 0$. Since all cycles containing these pairs contain two reverse pairs, the condition in Theorem~\ref{T1}(iii) is trivially satisfied.

As we prove in Appendix~\ref{app-ex:T1}, an \msc[] that has $\init[k]{}=0.4$ and positive off-diagonal transition probabilities $q_{ij}(\M) = q_{il}(\M) = q_{ji}(\M) = q_{jl}(\M) = q_{lj}(\M)=0.1$ rationalizes the stochastic choice data. Figure~\ref{fig:T1} visualizes the transitions occurring with positive probability. Note that this is the greatest number of possible direct comparisons consistent with the choice data. It is evident that this model is reversible, because it holds that 
$\q{ij}{}\q{jl}{}\q{li}{}=\q{ji}{}\q{lj}{}\q{il}{}$ and in all other cycles there are transitions happening with zero probability. 

Finally, we can verify that the rationalizing model is robust to \comp{}s. For example, if we consider a \comp{} with $r_{il}(\M)=r_{li}(\M)=0$, the model $\langle R(\M)\odot\mat[]{\M},\init[]{}\rangle$ also rationalizes the choice data. 
\label{ex:T1}
\end{example}

Comparability restrictions are assumed to have a uniform effect on all transition probabilities between non-restricted pairs. This assumption is plausible if the manipulation of the choice problem does not affect the accessibility and salience of non-restricted pairs. However, choices generated by reversible models are robust to an even larger class of \comp{}s. 

\begin{definition}
%A weak \comp{} is a mapping $r_w: [0,1]^{|\M|\times|\M|}\rightarrow [0,1]^{|\M|\times|\M|}$ such that if $r_w(\mat{\M})=\mat[']{\M}$ then for at least one pair $i,j\in\M$ holds $\q{ij}{}\neq \qp{ij}{}=\qp{ji}{}=0$, while preserving the communicating classes. For all remaining pairs $k,l\in\M$ holds $\q{kl}{}>0\implies \qp{kl}{}>0$ and $\q{kl}{}=0\implies \qp{kl}{}=0$.
%\end{definition}
A weak \comp{} of an \msc[] $\langle \mat[]{\M},\init[]{}\rangle$ is a symmetric matrix $R(\M)$ with non-negative off-diagonal elements such that $R(\M)\odot \mat{\M}$ is a right stochastic matrix with the same communicating classes as $\mat{\M}$.
\end{definition}
In essence, weak \comp{}s can have an asymmetric effect on non-restricted pairs. The symmetry of the matrix $R(\M)$ ensures TR-IIA holds for all pairs. 
 \begin{proposition}
  \label{prop:comparability1}
  If an \msc[] is reversible, the generated stochastic choice function $\bp[\Nset]$ is robust to all weak \comp{}s.
%  , therefore it holds for all \msc[]s $\mat[']{\M}$ that $\brho[]{}=\boldsymbol{\rho}(\mat[']{\M})$.
\end{proposition}
\begin{proof}
See Appendix~\ref{proof:comparability1}.
\end{proof}
In other words, agents who follow a reversible decision-making process cannot be nudged to choose certain alternatives more frequently by imposing comparability restrictions, even if these affect non-restricted pairs asymmetrically. Thus, the choices are robust across various presentations of the decision problem. For example, the choices of such decision makers are unaffected not only by reducing the number of possible comparisons, but also by reshuffling the positions of the items, leading to asymmetric changes in transition probabilities due to comparability restrictions. We conclude this subsection with the following example illustrating Proposition~\ref{prop:comparability1}. 
\addtocounter{example}{-1}
\begin{example}[continued]
Observe that an \msc[] in which the previously discussed comparability restriction between $i$ and $l$ affects asymmetrically the likelihood to make a transition between $i$ and $j$ such that $r_{ij}(\M)=2$, while keeping all other off-diagonal elements of the $R(\M)$ matrix equal to 1, will rationalize the same stochastic choice function.
\end{example}

\subsection{Pairwise and fully comparable MSC models}

%In this section we drop the IIA assumption of Proposition~\ref{prop:luce} and characterize \msc[]s in which the agent can directly compare all pairs of alternatives.
%%\input{input/figure/fig-complete}
In the previous section, we establish that choices generated by a reversible \msc[] cannot be manipulated with \comp{}s. Another type of decision maker unaffected by menu rearrangements is one whose cognitive abilities allow all pairwise comparisons. In this section, we characterize the choices of such decision makers, captured by classes of \msc[]s we call pairwise and fully comparable.

\begin{definition} An \msc[] is \textit{pairwise comparable} if it holds for all $i,j\in\M$ and all $\M\in\Nset$ that $\q{ij}{}=0$ implies $\q{ji}{}>0$. An \msc[] is \textit{fully comparable} if $\q{ij}{}>0$ for all $i,j\in\M$ and all $\M\in\Nset$. 
\end{definition}
This definition of pairwise comparability is a generalization of Assumption~\ref{a-pc} about pairwise comparability on binary sets to all menus. If the model is fully comparable, the agent makes all possible transitions with positive probability. Our characterization results of these classes of \msc[]s are based on the existence of particular sign-consistent cycles w.r.t.\ the stochastic choice function. 
\begin{definition}
The stochastic choice function $\bp[\M]$ is \btc over a pair of alternatives $(i,j)\in\M$ whenever if $\del[ij]{\M}\neq 0$, $\exists \cya[\M']$ with $(i,j)\in\cya[\M']$ and $\cya[\M']\in\cyrset[+/-]{\M}$ for $\M'\in\Mset$.
\end{definition}

Furthermore, the characterization of fully comparable \msc[]s requires in addition positivity. 

\begin{theorem}
\label{T2} 

The following are equivalent: A stochastic choice function $\bp[\M]$ is
	\begin{enumerate}
	\item[(i)] rationalizable by a pairwise (fully) comparable \msc[] on $\M$,
	\item[(ii)] (positive and) \btc over all pairs of alternatives.
%	\item[(ii)] (positive and) such that for all pairs $i,j\in\M$ with $\del[ij]{\M}\neq 0$, $\exists \cya[\M']$ with $(i,j)\in\cya[\M']$ and $\cya[\M']\in\cyrset[]{\M}$ for $\M'\in\Mset$.
	
%	such that for all pairs $i,j\in\M$ either there is no sign-consistent cycle with $(i,j)\in\cyr[]{\M'}$ or if there is a sign-consistent cycle with $(i,j)\in\cyr[]{\M'}$, it is such that for all $(k,l)\in\cyr[]{\M'}$,
%		$\p{k}{\M}\p[c]{l}{k,l}>\p{l}{\M}\p[c]{k}{k,l}.$
		
%	\[\frac{\p{k}{\M}}{\p{l}{\M}}\in  \left(\frac{\p[c]{k}{k,l}}{\p[c]{l}{k,l}}, r(\cy)\frac{\p[c]{k}{k,l}}{\p[c]{l}{k,l}}\right).\]
	\end{enumerate}	        
%
%A stochastic choice function \bp{} is rationalizable by an \msc[] such that
%
%(i) $\nexists\M\in\Mset$ and $i,j\in\M$ for which $\q{ij}{}=\q{ji}{}=0$ iff choice over all pairs and all menus is\btc{};
%
%(ii) $\q{ij}{}>0$ for all $i,j\in\M$ and $\M\in\Mset$ iff choice over all pairs and all menus is\btc{} and $\p[c]{i}{i,j}\in(0,1)$ for all $i,j\in\all$.

\end{theorem}
\begin{proof}
See Appendix~\ref{proof:T2}.
\end{proof}
Theorem~\ref{T2} implies that if a choice function violates condition (ii), the decision-making process necessarily took place under \comp{}s. What is more, the violation of the condition for an individual pair implies that no rationalizing model would assign a positive transition probability between the pair as we show with the next proposition. 
\begin{proposition}
\label{prop:no-edges}
A stochastic choice function $\bp[\M]$, which is not \btc over a pair $i,j\in\M$, is only rationalizable by \msc[]s with $\q{ij}{}=0$ .

%If for all proper cycles \cy with $(i,j)\in\cy$ there is some $(k,l)\in\cy$, for which $\p{k}{\M}\p[c]{l}{k,l}<\p{l}{\M}\p[c]{k}{k,l}$, then in all rationalizing \msc[]s $\q{ij}{}=\q{ji}{}=0$ .
%Then, in all rationalizing models $\q{ij}{}=\q{ji}{}=0$.
%. If a pair $i,j\in\M$ is not \btc then $\q{ij}{}=\q{ji}{}=0$.
\end{proposition}
	\begin{proof} See Appendix~\ref{proof:no-edges}.
	\end{proof}

We illustrate Theorem~\ref{T2} with the following example.

\begin{example}
Consider the choice function from binary sets given in Table~\ref{tab:example} and let $\bp[\set{i,j,k,l}] = \bp[\all]= (0.25,0.28,0.2,0.27)$. We first verify that $\bp[\all]$ satisfies the condition in Theorem~\ref{T2}(ii). Note that the cycle $\set{(i,k),(k,j),(j,l),(l,i)} \in \cyrset[+]{\all}$. We need to show that the remaining pairs $(i,j)$ and $(k,l)$ belong to a sign-consistent cycle. Since $\set{(i,k),(k,j),(j,i)}\in\cyrset[+]{\all}$ and $\set{(i,k),(k,l),(l,i)}\in\cyrset[+]{\all}$, the condition is satisfied. According to the theorem, the choice function is rationalizable by a fully comparable \msc[]. Appendix~\ref{app-ex:T2} contains an example of one such rationalizing model.  
\label{ex:T2}
\end{example}

%Next, we derive an analogous result to Theorem~\ref{T2} for fully comparable \msc[]s. Recall that a stochastic choice function is positive if $\p{i}{\M}>0$ for all $i\in\M$ and all $\M\in\Nset$. 
%\begin{corollary}
%\label{cor:T2-positive}
%The following are equivalent: A stochastic choice function $\bp[\M]$ is
%	\begin{enumerate}
%	\item[(i)] rationalizable by a fully comparable \msc[] on $\M$,
%	\item[(ii)] positive and such that for all pairs $i,j\in\M$ with $\del[ij]{\M}\neq 0$, $\exists \cya[\M']$ with $(i,j)\in\cya[\M']$ and $\cya[\M']\in\cyrset[]{\M}$ for $\M'\in\Mset$.
%%	\item[(iii)] positive and such that for all pairs $i,j\in\M'\subseteq\M$ with $\del[ij]{\M}\neq 0$, $\exists\cya[\M']$ such that $(i,j)\in\cya[\M']$ and for all $(k,l)\in\cya[\M']$ holds
%\begin{equation*}
%\label{eq:bc}
%%\frac{\p{k}{\M}}{\p{l}{\M}}\in  \left(\frac{\p[]{k}{\set{k,l}}}{\p[]{l}{\set{k,l}}}\prod_{(k',l')\in\cya[\M']}\frac{\p[]{l'}{\set{k',l'}}}{\p[]{k'}{\set{k',l'}}},\frac{\p[]{k}{\set{k,l}}}{\p[]{l}{\set{k,l}}}\right).
%\end{equation*} 
%	\end{enumerate}	        
%\end{corollary}
%\begin{proof}
%See Appendix~\ref{proof:T2-cor}.
%\end{proof}

%The above corollary implies that the choice probability ratio between each pair of alternatives from the larger menu is bounded by the choice probabilities from binary choice sets.
\subsection{Reversible and fully comparable MSC models}
		\label{sec:rev-characterization}
The classes of reversible and fully comparable \msc[]s are independent, because in a reversible \msc[] there might be a zero transition probability between some pairs (see Example~\ref{ex:T1}) and fully comparable models need not be reversible (see Example~\ref{ex:T2}). We now focus on the intersection of the two classes of models. Following Theorems~\ref{T1}(iii) and~\ref{T2}(ii), the generated stochastic choice function should be such that $\del[ij]{\M}=0$, or alternatively $\p[]{i}{\set{i,j}}\p{j}{\M}=\p[]{j}{\set{i,j}}\p{i}{\M}$, $\forall i,j\in\M$ and $\M\in \Nset$. This property is known as Independence of irrelevant alternatives (IIA) and implies that enlarging the choice set does not have an effect on relative choice probabilities. We state the characterization result below.
\begin{theorem}
\label{prop:luce}
The following are equivalent: A stochastic choice function $\bp[\M]$ is
	\begin{enumerate}
	\item[(i)] rationalizable by a fully comparable and reversible \msc[] on~$\M$,
	\item[(ii)] positive and independent of irrelevant alternatives,
	\item[(iii)] rationalizable by an \msc[] such that $\forall i,j \in\M$, $\frac{\q{ij}{}}{\q{ji}{}}=\frac{u(j)}{u(i)}$, where $u:\all\rightarrow \mathbb{R}_{++}$.
	\end{enumerate}
\end{theorem}
\begin{proof}
	See Appendix~\ref{proof:luce}.
\end{proof}
We gain several important insights from Theorem~\ref{prop:luce}. First, we learn that the well-known Luce model, also referred to as multinomial logit, is observationally equivalent to a reversible and fully comparable \msc[]. Recall that a stochastic choice function~$\bp[\Nset]$ is a Luce rule if there exists a function ${u:\all\rightarrow \mathbb{R}_{++}}$ such that for all $\M \in \Nset$ and $i\in \M$
\[\p{i}{\M}=\frac{\val{i}}{\sum_{j\in \M} \val{j}}. \]
In his seminal work, \citet{Luce1959} shows that positivity and IIA characterize the Luce model. Therefore, Theorem~\ref{prop:luce} provides a procedural justification for the Luce model. Basically, stochastic choice consistent with the Luce model arises whenever the search dynamics follow a Markov process with symmetric average transitions and positive transition probabilities between all pairs of alternatives.

Therefore, choices consistent with reversible and fully comparable models can be considered most unbiased, as all pairwise comparisons are possible, determined solely by utility, and unaffected by the presentation of the choice problem. Moreover, if the choice data satisfies Theorem~\ref{T1}(iii), but violates IIA, we can deduce that the decision making process is reversible but the agent is unable to make all pairwise comparisons.

Finally, Theorem~\ref{prop:luce}(iii) provides an interesting insight in the exploration mechanism that corresponds to a reversible and fully comparable \msc[]. Specifically, the ratio of transition probabilities between each pair of alternatives equals the ratio of their utilities, where the utility function is the same as in the Luce model. This feature suggests that transitions between alternatives depend solely on the utility, not on factors like salience, positioning, or cognitive restrictions. In contrast, violations of IIA and positivity indicate that these factors do influence decisions, making the agent susceptible to choice manipulation through the presentation of the decision problem. 
%Therefore, in our modelling framework the agents following a reversible and fully comparable model can be thought of as being rational agents, although they still make stochastic choices. 

%% file: input/figure/fig-shelf.tex
	\begin{figure}[h!]
\centering
    \begin{subfigure}{0.45\textwidth}
        \centering
        \begin{tikzpicture}[->,>=stealth',shorten >=1pt,auto,node distance=4cm]
			\begin{scope}[every node/.style={circle,draw,minimum size = 0.8cm}]
			    \node (i) at (0,0) {$i$};
			    \node (j) at (2.5,0) {$j$};
			    \node (l) at (0,2.5) {$l$};
			    \node (k) at (2.5,2.5) {$k$};
			\end{scope}
			\begin{scope}[]
			    \path [<->] (i) edge node {} (j);
			    \path [<->] (j) edge node {} (k);
			    \path [<->] (k) edge node {} (l);
			    \path [<->] (l) edge node {} (i);
			    \draw (j) to[out=25,in=-25,looseness=4] (j); 
	    		\draw (i) to[out=155,in=205,looseness=4] (i); 
	    		\draw (k) to[out=25,in=-25,looseness=4] (k); 
				\draw (l) to[out=155,in=205,looseness=4] (l);
			\end{scope}
			\end{tikzpicture}
    \end{subfigure}\hfill
    \begin{subfigure}{0.45\textwidth}
        \centering
        \begin{tikzpicture}[->,>=stealth',shorten >=1pt,auto,node distance=4cm]
			\begin{scope}[every node/.style={circle,draw,minimum size = 0.8cm}]
			    \node (i) at (0,0) {$i$};
			    \node (j) at (2.5,0) {$j$};
			    \node (k) at (0,2.5) {$k$};
			    \node (l) at (2.5,2.5) {$l$};
			\end{scope}
			\begin{scope}[]
			    \path [<->] (i) edge node {} (j);
			    \path [<->] (j) edge node {} (l);
			    \path [<->] (l) edge node {} (k);
			    \path [<->] (k) edge node {} (i);
				\draw (j) to[out=25,in=-25,looseness=4] (j); 
	   			\draw (i) to[out=155,in=205,looseness=4] (i); 
				\draw (l) to[out=25,in=-25,looseness=4] (l); 
				\draw (k) to[out=155,in=205,looseness=4] (k);
			\end{scope}
			\end{tikzpicture}
    \end{subfigure}
    \caption{Limitations in a decision maker's perception can lead to varying pairwise comparability restrictions depending on the item positioning.}
    \label{fig:shelf}
\end{figure}
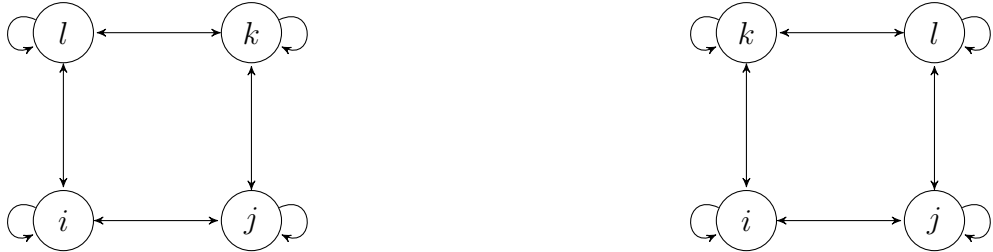

%% file: input/table/tab-T1.tex
\begin{table}
\begin{center}
\caption{Choice probabilities of alternatives $i,j,k,l$ (rows) from binary choice sets (columns) used in Examples~\ref{ex:T1}-\ref{ex:T3}.\label{tab:example}}
\begin{tabular}{ c| c c c c c c}
\hline
 	& $\{i,j\}$ & $\{i,k\}$ & $\{i,l\}$ & $\{j,k\}$ & $\{j,l\}$ & $\{k,l\}$ \\ 
 	\hline
$i$	& 0.5		& 0.5 		& 0.5 		&			&			&			\\  
$j$ & 0.5		&			& 			& 0.6		& 0.5		& 			\\   
$k$	& 			& 0.5		&			& 0.4		&			& 0.4		\\ 
$l$	& 			& 			& 0.5		&			& 0.5		& 0.6 		\\
\hline
\end{tabular}
\end{center}
\end{table}

%% file: input/figure/fig-T1.tex
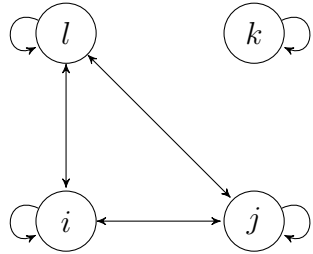
\begin{figure}[ht!]
\centering
    \begin{tikzpicture}[->,>=stealth',shorten >=1pt,auto,node distance=4cm]
	\begin{scope}[every node/.style={circle,draw,minimum size = 0.8cm}]
	    \node (i) at (0,0) {$i$};
	    \node (j) at (2.5,0) {$j$};
	    \node (l) at (0,2.5) {$l$};
	    \node (k) at (2.5,2.5) {$k$};
	\end{scope}
	\begin{scope}[]
	    \path [<->] (i) edge node {} (j);
	    \path [<->] (l) edge node {} (i);
	    \path [<->] (l) edge node {} (j);
	    \draw (j) to[out=25,in=-25,looseness=4] (j); 
	    \draw (i) to[out=155,in=205,looseness=4] (i); 
	    \draw (k) to[out=25,in=-25,looseness=4] (k); 
		\draw (l) to[out=155,in=205,looseness=4] (l);
	\end{scope}
	\end{tikzpicture}    
    \caption{A reversible \msc[] rationalizing the choice function given in Example~\ref{ex:T1}.}
    \label{fig:T1}
\end{figure}

%% file: input/text/4-first-fixation.tex
\section{Initial fixation}
\label{sec:fixation}
Apart from comparability restrictions, the presentation of the decision problem affects the exploration and evaluation process through the first fixation, captured by the initial distribution $\init[]{}$. It is evident by the form of the generated choice function given in~\eqref{msc-b} that our model predicts that the initial fixation can have significant and long-lasting effects on the final choices when there is time pressure, which is also in line with existing experimental evidence \citep{Reutskaja2011, Armel2008, Atalay2012}. Thus, a choice problem designer could steer the choices of such decision makers towards a target alternative by influencing the initial fixation probability.\footnote{In order to ensure the effectiveness of such an intervention our model shows that it is not only necessary that initial fixation probability of the target increases, but also that the fixation probabilities of all other competitors weakly decrease. A detailed proof is available upon request.} Without time constraints, the effectiveness of such interventions is more nuanced and depends on the pairwise comparability.

%\subsection{Decision making without time constraints}
Following Proposition~\ref{prop:limit}, when the pairwise evaluation process can eventually reach all alternatives in the menu from any starting point, with no limit on the number of transitions (indicating an irreducible Markov process\footnote{A fully comparable \msc[] is an example of an irreducible model, whereas the model depicted in Figure~\ref{fig:T1} is reducible -- if the process starts at alternative~$k$, none of the other alternatives in the menu can be reached and vice versa.}), final choices become independent of the initial fixation distribution, \ie the choice function is robust to any initial fixation distributions. 

We now characterize irreducible \msc[]s. In order to do so, we need to relax appropriately the property in Theorem~\ref{T2}(ii) since pairwise comparable models are a subclass of irreducible \msc[]s.

\begin{theorem}
\label{T3} 
The following are equivalent: A stochastic choice function \bp[\M] is 
	\begin{enumerate}
        \item[(i)] rationalizable by an irreducible \msc[] on $\M$,
        \item[(ii)] robust to any initial fixation probability distributions,
%		\item[(iii)] positive and such that $\exists \cya[\M]$ with at least $|\cya[\M]|-1$ pairs $(i,j)\in \cya[\M]$ for which $\p[]{i}{\set{i,j}}>0$ and , then $\exists \cya[\M']$ with $(i,j)\in\cya[\M']$ and $\cya[\M']\in\cyrset[+/-]{\M}$ for some $\M'\in\Mset$.	
%		\item[(iii)] positive and for all pairs $(i,j)\in\M$, there exists $\cya[\M']$ with $(j,i)\in\cya[\M']$ and $\M'\in\Mset$ for which it is \btc over all $(k,l)\in\cya[\M']\setminus (j,i)$.
		\item[(iii)] \btc over all pairs of some cycle $\cya[\M]$.

%		, where for all $(k,l)\in\cya[\M']\setminus (i,j)$ holds that if $\del[kl]{\M}\neq 0$, $\exists \cya[\M'']$ with $(k,l)\in\cya[\M'']$ and $\cya[\M'']\in\cyrset[+/-]{\M}$ for some $\M', \M''\in\Mset$.

%		such that $\exists S\subset \cya[\M]$ with $|S|=|\cya[\M]|-1$ such that for all pairs $(i,j)\in S$ with $\del[ij]{\M}\neq 0$ and $\p[]{i}{\set{i,j}}>0$, $\exists \cyr[]{\M'}{\M}$ with $(i,j)\in\cyr[]{\M'}{\M}$.	
	%	either $(i,j)\in\cyr[]{\M'}$ or for all $\cya[]{\M'}$ with $(i,j)\in\cya[]{\M'}$ holds $\del[kl]{}=0$ for all $(k,l)\in\cya[]{\M'}$ and $\M'\in\Mset$.
		
		%there is no proper cycle with $(i,j)\in\cyr[]{\M'}$ or if there is a proper cycle with $(i,j)\in\cyr[]{\M'}$, it is such that for all $(k,l)\in\cyr[]{\M'}$,
		%$\p{k}{\M}\p[c]{l}{k,l}>\p{l}{\M}\p[c]{k}{k,l}.$           
	\end{enumerate}
\end{theorem}
\begin{proof}
	See Appendix~\ref{proof:T3}.
\end{proof}

The reducibility of the decision making process implies the existence of separate communicating classes, or consideration sets. This aligns with an agent's approach of first focusing on a subcategory of alternatives and then selecting an item from that subcategory without exploring the rest. Such partitions of the menu could stem from the decision maker's subjective perception of item attributes, such as taste (sweet vs.\ salty), brand, or ingredients. This behavior is consistent with experimental marketing literature, which identifies heuristics used by decision makers to construct consideration sets by focusing on subsets of items with specific attributes (see \citet{Hauser2014} and references therein).

The partitioning of a menu into subsets can also be influenced externally  by the presentation of the decision problem. For instance, in a menu composed of snacks, a designer can emphasize the division between salty and sweet snacks by placing them on different shelves in a supermarket. On an {e-commerce} platform with a recommendation system, the algorithm might categorize items into distinct subgroups and suggest different sets of related products based on the first item the user views. 

If the decision making process is reducible, \ie the underlying Markov chain has multiple communicating classes, the first fixation determines the subset of reachable alternatives. Thus, choices can be manipulated through the initial distribution, unless the probability to start in each communicating class remains constant. Increasing the likelihood of starting in a particular communicating class boosts the choice probability of all alternatives in that class while preserving their relative choice probabilities within the class. Thus, a social planner can nudge decision makers to choose a target alternative more often by directing attention to any alternative in the same communicating class as the target.

Therefore, inferring the irreducibility property of an \msc[] from the choice function is advantageous in two ways. It helps evaluate the effectiveness of various attention-grabbing methods in increasing the choice probability of a target alternative, and it reveals whether the menu is perceived as fragmented or coherent.

Note that if a stochastic choice function fulfils Theorem~\ref{T3}(iii), this does not necessarily mean that it was generated by an irreducible model, but only that it is consistent with such a model. An analyst can infer the reducibility properties of the true model by observing changes in final choice probabilities resulting from alterations in the initial fixation probability of different communicating classes. In order to do that effectively, the analyst would ideally know which alternatives belong to the same communicating class. As we showed in Proposition~\ref{prop:no-edges}, the stochastic choice function itself provides insights about the possible pairwise comparisons and thus about the composition of the communicating classes. Thus, Theorem~\ref{T3}(ii) and (iii) can complement each other in order to infer the irreducibility property.

If choice data reveals the existence of consideration sets when applying Theorem~\ref{T3}, sellers can infer which products fall into the same consideration set as their own product, thus inferring their direct competitors from consumers' perspective. Furthermore, this insight aids choice problem designers in making better targeted interventions and predicting the effects of a change in the initial fixation probability. We conclude this section with the following example.
%\footnote{Note also that the choice behavior of the agent reveals when the only rationalizing model is one where no transitions between alternatives are allowed, namely $\q{ij}{}=0$ for all $i,j\in\M$. According to such a model the agent chooses the alternative that he first views. This is the only rationalizing \msc[] if $\spo$ is complete and acyclical.}
%We illustrate the properties of the stochastic choice function defined above with the following example.

\input{input/figure/fig-T3}
\begin{example}
Consider the choice function from binary sets given in Table~\ref{tab:example}. Let the choice function from the menu $\M=\{i,j,k,l\}$ be $\bp[\M] = (0.24,0.3,0.22,0.24)$. We first verify that $\bp[\M]$ satisfies Theorem~\ref{T3}(iii). Observe that $\del[il]{\M}=0$ and consider the cycle $\cya[\M']=\set{(j,i),(i,k),(k,j)}\in\cyrset[+]{\M}$. Thus, the cycle $\cya[\M] = \set{(j,i),(i,l),(l,i),(i,k),(k,j)}$ fulfils the requirement and the stochastic choice function is rationalizable by an irreducible \msc[] such as the one contained in Appendix~\ref{app-ex:T3} and depicted in Figure~\ref{fig:T3}. 

We can verify that the rationalizing model from Figure~\ref{fig:T3} has the highest number of possible transitions using Proposition~\ref{prop:no-edges}. In order to show that any rationalizing model assigns zero probability to transitions between $j$ and $l$, note first that $\del[jl]{\M}>0$, so we need to check whether the pair belongs to any sign-consistent cycle according to \bp[\M]. We consider the cycles $\{(j,k),(k,l),(l,j)\}$ and $\{(k,l), (l,j),(j,i),(i,k)\}$, because the pair cannot be part of a sign-consistent cycle with $(j,l)$ since $\del[il]{\M}=0$. We can easily see that these are no sign-consistent cycles since $\del[lj]{\M}<0$ and $\del[kl]{\M}>0$. We can use the same argument for the pair $(k,l)$ to prove that the transition probability should be zero.
\label{ex:T3}
\end{example}

%% file: input/figure/fig-T3.tex
\begin{figure}[ht!]
\centering
    \begin{tikzpicture}[->,>=stealth',shorten >=1pt,auto,node distance=4cm]
	\begin{scope}[every node/.style={circle,draw,minimum size = 0.8cm}]
	    \node (i) at (0,0) {$i$};
	    \node (j) at (2.5,0) {$j$};
	    \node (l) at (0,2.5) {$l$};
	    \node (k) at (2.5,2.5) {$k$};
	\end{scope}
	\begin{scope}[]
	    \path [<->] (i) edge node {} (j);
	    \path [<->] (j) edge node {} (k);
	    \path [<->] (k) edge node {} (i);
	    \path [<->] (i) edge node {} (l);
	    \draw (j) to[out=25,in=-25,looseness=4] (j); 
	    \draw (i) to[out=155,in=205,looseness=4] (i); 
	    \draw (k) to[out=25,in=-25,looseness=4] (k); 
		\draw (l) to[out=155,in=205,looseness=4] (l);	
	\end{scope}
	\end{tikzpicture}
    \caption{An irreducible \msc[] rationalizing the choice function given in Example~\ref{ex:T3}.}
    \label{fig:T3}
\end{figure}
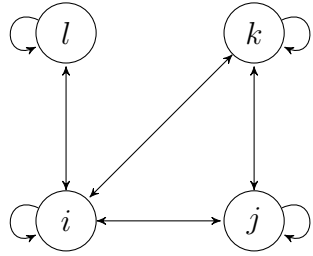

%% file: input/text/5-literature.tex
\section{Related literature}
\label{sec:literature}
Researchers across various disciplines have recognised the advantages of incorporating Markov processes in modeling choice behavior. The literature emphasizes the tractability of these models, their superior performance over traditional frameworks in predicting choices, and their ability to generalize popular random utility-based discrete choice models, such as the multinomial logit model \citep{Ragain2016,Blanchet2016,Cerreia-Vioglio2018,Baldassi2020}. We relate to this literature by employing a general Markovian model to study the effects of choice architecture and to reveal unobservable pairwise comparison patterns from choice behavior.

Furthermore, our characterization results make a technical contribution related to the problem of inverting the stationary distribution of a Markov chain studied in \citet{Morimura2013} and \citet{Kumar2015}. This existing literature attempts to recover the compete transition probability matrix from the stationary distribution assuming that either the graph or the rate of reaching one state from another is (partially) known. We identify the properties of a stationary distribution that guarantee the existence of a corresponding Markov chain within the classes of reversible, pairwise comparable, and irreducible chains. Notably, we achieve this under minimal assumptions concerning the consistency of Markov chains over subsets of states.

The \msc[] builds on a vast theoretical literature incorporating various behavioral limitations into models of choice behavior, in particular to search, reference dependence, and limited consideration. In much of the theoretical literature featuring sequential menu exploration, the search order is specified by the choice problem itself in a form of a list \citep{Horan2010, Papi2012, Zhang2016, Rubinstein2006} or unobservable network of related products in a recommender system \citep{Masatlioglu2013, Masatlioglu2019}. In some choice models such as those by \citet{Caplin2011}, \citet{Apesteguia2013}, \citet{Dutta2020} and \citet{Demirkan2024} agent's ability to partition and explore the choice set is not observable and this literature is mainly concerned with the problem of rationalizability. Our \msc[] adds to this literature by incorporating many of the empirically documented decision-making patterns motivating the existing models and offering a simple framework to analyse the effects of various presentation formats.  

Reference points have played an important role in behavioral choice models, assumed to be static and exogenous \citep{Tversky1991, Masatlioglu2005, Salant2008}, endogenous and menu-dependent \citep{Ok2015, Tserenjigmid2019}, or stochastic \citep{Kibris2024, Ravid2019}. The \msc[] introduces a dynamic reference point represented by the current alternative, since the probability of transitioning to another option is influenced by the present state in the Markov chain. Alternatively, the initial fixation can be interpreted as a stochastic reference point. Our contribution to this literature lies in highlighting the nuanced effect that reference points have on final choices, depending on the time pressure and the existence of consideration sets. 
%
%\citet{Kovach2018}
%A feature that the model shares with the \msc[] is the fact that the probability to pay attention to an alternative depends on the reference point.

Limited consideration in decision-making, \ie the idea that agents do not consider all available alternatives, has also been a very influential concept in behavioral choice literature. Consideration sets have been modeled as exogenous \citep{masatlioglu2012,demuynck2018}, random \citep{masatlioglu2012,Manzini2014a, Brady2016}, or as a result of constrained optimization \citep{caplin2018}. We contribute to this literature by identifying necessary and sufficient conditions showing the existence of such consideration sets and their composition despite the fact that the agent makes stochastic choices even if all alternatives are considered. 

Finally, our paper is related to the literature on manipulating attention by different presentations of the menu and its consequences on choice behavior. Notable recent theoretical works addressing this issue are \citet{Gossner2021}, which features a learning mechanism and endogenous stopping, and \citet{Kovach2019}, which enriches the \citet{Luce1959} model by separating the alternatives into two categories depending on their ability to attract attention. Our framework enables us to study in addition the impact of manipulating the comparability between pairs of alternatives on final choices.

%% file: input/text/6-conclusion.tex
\section{Conclusion}
\label{sec:conclusion}
Motivated by key behavioral patterns such as pairwise comparisons, attention-driven transitions, and stochastic decision-making, we introduce the Markov Stochastic Choice (MSC) model as a tool to analyse the effects of item arrangement and attention manipulation on final choices. Understanding and predicting these effects is crucial for choice problem designers interested in influencing or minimizing external effects on choices.

Our findings indicate that decision-makers may differ in their sensitivity to presentation formats. We identify decision-makers following reversible or fully comparable Markov processes as being unaffected by item positioning and show how these properties can be revealed from choice data. Additionally, we provide a behavioral foundation for the Luce model as observationally equivalent to a reversible and fully comparable \msc[] without time constraints. Finally, interventions such as altering initial fixation probabilities can effectively nudge decisions, especially under time constraints or when the decision process is reducible. We are able to determine the existence and composition of consideration sets from the choice function.

The proposed model offers different directions for future research, two of which we outline here. In the current model, time pressure is reflected in the stopping probability of the decision-making process. An alternative way to model time pressure is as a fixed point in time at which the exploration of the choice set is terminated. One could then compare the effect that the modelling of the two approaches to time pressure have on final choices and use that information to make predictions.    

The second direction is to extend the model to multi-attribute alternatives. This can be easily incorporated in the existing framework by letting each state of the Markov chain represent an attribute of an alternative. The multi-attribute version of the \msc[] offers a natural environment to study the other context effects known in the literature, such as the compromise and similarity effects, and contribute to a better understanding of the effect of alternative- vs.\ attribute-based transitions on choices, which have been extensively studied in the experimental literature (see \citet{Noguchi2014}).

%% file: input/text/app-2-proof-limit.tex
\subsection{Proof of Proposition~\ref{prop:limit}}
\label{proof:limit}
We begin the proof by simplifying Equation~\eqref{msc-b}.\footnote{An alternative way of defining an \msc[] is an absorbing Markov chain with one one transient and one absorbing state corresponding to each alternative in the menu and transition probability matrix between transient states given by $\mat[*]{\M}=(1-\abs)\mat[]{\M}$ and absorption probability matrix $A(\M)=\abs I$. The stationary distribution of this Markov chain is given by ${\bra[]{\init[]{*}}{\mat[*]{\M}}=\init[]{*}(I-\mat[*]{\M})^{-1}A(\M)}$ \citep{abs_chains}. Thus, Lemma~\ref{def} is a corollary of existing results from the literature on absorbing Markov chains. Nevertheless, the proof is provided here for convenience.} 
\vspace{1pt}
\begin{lemma} 
\label{def}
The generated stochastic choice function in~\eqref{msc-b} is equivalent to
 	\begin{equation}
	\bra[]{\init[]{}}{\mat{\M}} = \abs \init[]{} (I-(1-\abs)\mat[]{\M})^{-1}, \quad \forall \M\in\Nset.
	\label{eq-def}
	\end{equation}
 \end{lemma}
\begin{proof}
	Let $\tmat{\M} = (1-\abs)\mat[]{\M}$. Then, $
	\bra[]{\init[]{}}{\mat{\M}} =\abs \init[]{} \sum_{\period=0}^\infty \tmat{\M}^{\period}.$
	In order to simplify the expression, we need to show that the Neumann series $\sum_{\period=0}^\infty \tmat{\M}^{\period}$ converges in the operator norm in order to use that $\sum_{\period=0}^\infty \tmat{\M}^{\period}=(I-\tmat{\M})^{-1}$ (see for example \citet[Theorem 2.14]{Kress2014}). This is the case if $\tmat{\M}$ is smaller than unity in some norm. Since $\mat[]{\M}$ is a stochastic matrix $||\mat[]{\M}||_{\infty} = \max_i \sum_j|\q{ij}{\M}|=1$ and $||\tmat{\M}||_{\infty} = 1-\abs <1$. Hence, ${\bra[]{\init[]{}}{\mat{\M}} = \abs \init[]{} (I-\tmat{\M})^{-1}}$.
\end{proof}

Fix a menu $\M\in\Mset$. We reformulate Equation~\eqref{eq-def} in the following way:
\begin{align*}
\bra[]{\init[]{}}{\mat[]{\M}}
%&= \frac{\abs}{1-\abs} \init[]{} \left(\frac{1}{1-\abs}I-\mat[]{\M}\right)^{-1} \\
&=\frac{\abs}{1-\abs} \init[]{} \frac{1}{\text{det}\left(\frac{1}{1-\abs}I-\mat[]{\M}\right)}\text{adj}\left(\frac{1}{1-\abs}I-\mat[]{\M}\right),
\end{align*}
where adj(.) denotes the adjugate matrix. Since the Markov process is ergodic, $\mat[]{\M}$ has an eigenvalue of 1. We express the determinant with the characteristic polynomial:
\begin{align*}
\text{det}\left(\frac{1}{1-\abs}I-\mat[]{\M}\right)&=\prod_{l=1}^L \left(\frac{1}{1-\abs} -\lambda_l\right)=\frac{\abs}{1-\abs}\prod_{l=2}^L \left(\frac{1}{1-\abs} -\lambda_l\right),
%&=\left(\frac{1}{1-\abs} -1\right)\prod_{l=2}^L \left(\frac{1}{1-\abs} -\lambda_l\right)
\end{align*}
where $\lambda_l$ denotes the eigenvalues of $\mat[]{\M}$. Since $\lambda_l< 1$ for $l\neq 1$, the determinant is positive. We plug in the expression and simplify the resulting choice function
\begin{align*}
\bra[]{\init[]{}}{\mat[]{\M}} &= \init[]{} \frac{1}{\prod_{l=2}^L \left(\frac{1}{1-\abs} -\lambda_l\right)}\text{adj}\left(\frac{1}{1-\abs}I-\mat[]{\M}\right).
\end{align*}
Now we can let $\abs\rightarrow 0$:
\begin{align*}
\lra[]{\init[]{}}{\mat[]{\M}}=\lim_{\abs\rightarrow 0}\bra[]{\init[]{}}{\mat[]{\M}}&= \init[]{} \frac{1}{\prod_{l=2}^L (1 -\lambda_l)}\text{adj}(I-\mat[]{\M}).
\end{align*}
Multiply both sides of the equation with $(I-\mat[]{\M})$ from the right and simplify:
\begin{align*}
\lra[]{\init[]{}}{\mat[]{\M}}(I-\mat[]{\M})
%&= \init[]{} \frac{1}{\prod_{l=2}^L (1 -\lambda_l)}\text{adj}(I-\mat[]{\M})%(I-\mat[]{\M})\\
&= \init[]{} \frac{1}{\prod_{l=2}^L (1 -\lambda_l)}\text{det}(I-\mat[]{\M}).
\end{align*}
Since $\text{det}(I-\mat[]{\M})=0$, we have $\brho[]{}(I-\mat[]{\M})=\boldsymbol{0}$. 
%Hence, $\bra[]{\init[]{}}{\mat[]{\M}}$ converges to the unique stationary distribution of the Markov chain with transition probability matrix $\mat[]{\M}$ as $\alpha \rightarrow 0$.

%% file: input/text/app-3-proof-auxiliary.tex
\subsection{Auxiliary results}
\label{sec:aux}
This section contains three key lemmas for proving the main theorems.
\label{proof:auxiliary}
%Let $\dif{ji}{}=\p{j}{\M}-\frac{\p[]{j}{\set{i,j}}}{\p[]{i}{\set{i,j}}}\p{i}{\M}$ for $\p[]{i}{\set{i,j}}\neq 0$. Therefore, we define $i\spo j$ whenever $\dif{ji}{}>0, \forall i,j\in\M.$ Here we show some auxiliary results which will be useful in the proofs of Theorems \ref{T1}-\ref{T3}.

\begin{lemma} 
\label{lem1}
Let $\bp[\Nset]$ be a stochastic choice function rationalizable by an \msc[]. For all $\M\in\Nset$ and $i,j\in\M$ for which $\p[]{i}{\set{i,j}}\in (0,1)$ it holds that $$\frac{\del[ij]{\M}\q{ij}{}}{\p[]{j}{\set{i,j}}}=-\frac{\del[ji]{\M}\q{ji}{}}{\p[]{i}{\set{i,j}}}.$$
\end{lemma}
\begin{proof} Take an arbitrary pair $i,j\in\M$ with $\p[]{i}{\set{i,j}}\in (0,1)$. Rationalizability of $\bp[\Nset]$ implies $\p[]{i}{\set{i,j}}\q[c]{ij}{i,j}=\p[]{j}{\set{i,j}}\q[c]{ji}{i,j}$. It follows from $\p[]{i}{\set{i,j}}\in (0,1)$ and Assumption~\ref{a-pc} that $\q[c]{ij}{i,j}>0$ and $\q[c]{ji}{i,j}>0$. The statement is trivially satisfied for $\q{ij}{}=\q{ji}{}=0$. If $\q{ij}{}>0$, TR-IIA implies $\q{ji}{}>0$ and $\frac{\q{ij}{}}{\p[]{j}{\set{i,j}}}=\frac{\q{ji}{}}{\p[]{i}{\set{i,j}}}$.
Therefore,
	\begin{align*}
	\frac{\del[ij]{\M}\q{ij}{}}{\p[]{j}{\set{i,j}}}=&\left(\p{i}{\M}\p[]{j}{\set{i,j}}-\p[]{i}{\set{i,j}}\p{j}{\M}\right)\frac{\q{ji}{}}{\p[]{i}{\set{i,j}}},
	\end{align*}
	and the final result follows.
\end{proof} 
\begin{lemma}
\label{lem3}
A stochastic choice function $\bp[\Nset]$ is rationalizable with an \msc[] iff
	\begin{equation}
	\label{eq:diff}
	\sum\limits_{\set{i\in\M:\q{ji}{} >0}}\frac{\del[ji]{\M}\q{ji}{}}{\p[]{i}{\set{i,j}}}-\sum_{\substack{\set{i\in\M:\q{ji}{} =0,\\ \q[c]{ji}{i,j}=0}}}\frac{\del[ij]{\M}\q{ij}{}}{\p[]{j}{\set{i,j}}}=0\quad \forall j \in \M, \forall \M\in\Nset.
	\end{equation}
\end{lemma}
\begin{proof}
	The rationalizability of $\bp[\Nset]$ implies that $\bp[\M](I-\mat[]{\M}) = 0$ for all $\M\in\Nset$. Therefore, it holds for all $j \in \M$ and $\M\in\Nset$ that
		\begin{equation*}
		\begin{gathered}
		\sum_{i\neq j}\p{j}{\M}\q{ji}{}-\sum_{i\neq j}\p{i}{\M}\q{ij}{}=0.
		\end{gathered}
		\end{equation*}
		We now rearrange the terms depending on whether $\q{ji}{} >0$ or not.
		\begin{equation*}
		\begin{gathered}
		\sum\limits_{\set{i\in\M\setminus\set{j}:\q{ji}{} >0}}\left(\p{j}{\M}-\p{i}{\M}\frac{\q{ij}{}}{\q{ji}{}}\right)\q{ji}{}-\sum\limits_{\substack{\set{i\in\M\setminus\set{j}:\\ \q{ji}{} =0}}}\p{i}{\M}\q{ij}{}=0
		\end{gathered}
		\end{equation*}
		It follows from Assumptions~\ref{a-pc} and~\ref{a-TR-IIA} that if $\q{ji}{}\neq 0$, then $\q[c]{ji}{i,j}\neq 0$. Rationalizability and Assumption~\ref{a-pc} imply that $\p[]{i}{\set{i,j}}\neq 0$.
		We apply Assumption~\ref{a-TR-IIA} and the rationalizability of the choice function from binary sets to obtain:
		\begin{equation*}
		\begin{gathered}
%		\sum\limits_{\set{i\in\M\setminus\set{j}:\q{ji}{} >0}}\left(\p{j}{\M}-\p{i}{\M}\frac{\q[c]{ij}{i,j}}{\q[c]%{ji}{i,j}}\right)\q{ji}{}-\sum\limits_{\set{i\in\M\setminus\set{j}:\q{ji}{} =0}}\p{i}{\M}\q{ij}{}=0\\
		\sum\limits_{\substack{\set{i\in\M\setminus\set{j}:\\ \q{ji}{} >0}}}\left(\p{j}{\M}-\p{i}{\M}\frac{\p[]{j}{\set{i,j}}}{\p[]{i}{\set{i,j}}}\right)\q{ji}{}-\sum\limits_{\substack{\set{i\in\M\setminus\set{j}:\\ \q{ji}{} =0}}}\p{i}{\M}\q{ij}{}=0,
		\end{gathered}
		\end{equation*}
	Note that TR-IIA and $\q{ji}{}= 0$ imply that either $\q{ij}{}= 0$ or $\q[c]{ji}{i,j}= 0$. In the latter case rationalizability guarantees that $\p{i}{\set{i,j}}= 0$. Hence, for those pairs holds $\del[ij]{\M}=\p{i}{\M}\p[]{j}{\set{i,j}} -\p[]{i}{\set{i,j}}\p{j}{\M}= \p{i}{\M}$ and the final result follows.
		\end{proof}

Let $\G=\set{(i,j)\in\M^2:\del[ij]{\M}>0} = \set{g_b^{\M}}_{b\in\set{1,\dots,|\G|}}$ be an indexed set of ordered pairs. The menu is also denoted as an indexed set $\M=\set{m_a:a\in\set{1,\dots,|\M|}}$. 
%
% bijective mapping on the set of ordered pairs. Further, let $h_{\M}:\M\rightarrow[1,\dots , |\M|]$ be a bijective mapping assigning an index to each alternative in the menu $\M$. 
\begin{lemma}
\label{lem4}
A stochastic choice function $\bp[\Nset]$ is rationalizable with an \msc[] iff 
\begin{equation}
	\label{eq:diff-system}
	\mathcal{D}(\M)\sav{}=\boldsymbol{0},\quad \forall \M\in\Nset.
	\end{equation}
	where $\sav$ is a vector of transition probabilities between the pairs in $\G$, \ie $\gamma_{b}(\M)=\q{ij}{}$ iff $g_{b}^M=(i,j)$ and $\mathcal{D}(\M)$ is a matrix with dimensions $|\M|\times|\G|$ and elements
 \begin{align*}
\mathrm{d}_{a,b}(\M)=\begin{cases}
\del[ij]{\M}/\p[]{j}{\set{i,j}}, &\text{ if } m_{a}=i,g_b^M=(i,j)\\
-\del[ij]{\M}/\p[]{j}{\set{i,j}}, &\text{ if } m_{a}=j,g_b^M=(i,j),\\
0 &else.
\end{cases}
\end{align*}
\end{lemma}
\begin{proof} Take an arbitrary menu $\M\in\all$.
Observe that the linear system~\eqref{eq:diff} in Lemma~\ref{lem3} is trivially satisfied whenever \bp[\M] is such that $\del[ij]{\M}=0$ for all pairs $i,j\in\M$. If it holds for at least one pair $i,j\in\M$ that $\del[ij]{\M}>0$, then $\G\neq \emptyset$, which we assume for the remainder of the proof.

Note that all elements of the matrix $\mathcal{D}(\M)$ are well-defined since $\del[ij]{\M}>0$ implies $\p[]{j}{\set{i,j}}>0$. Further, none of the columns of the matrix is equal to $\textbf{0}$ because each column corresponds to a pair $(i,j)$ for which $\del[ij]{\M}>0$. 

We claim that the linear system~\eqref{eq:diff-system} is equivalent to the one in~\eqref{eq:diff} and is thus necessary and sufficient for rationalizability. Consider an equation from the linear system~\eqref{eq:diff-system} corresponding to an arbitrary matrix row $a$ such that $m_{a}=j$:
	\begin{equation*}
	\begin{gathered}
\sum\limits_{\set{b:g_b^M=(j,i), \forall i\in\M}}\mathrm{d}_{a,b}(\M)\gamma_{b}(\M)+ \sum\limits_{\set{b:g_b^M=(i,j), \forall i\in\M}}\mathrm{d}_{a,b}(\M)\gamma_{b}(\M)=0\\
 \sum\limits_{\set{i\in\M:(j,i)\in \G}}\frac{\del[ji]{\M}\q{ji}{}}{\p{i}{\set{i,j}}}-\sum\limits_{\set{i\in\M:(i,j)\in \G}}\frac{\del[ij]{\M}\q{ij}{}}{\p{j}{\set{i,j}}}=0\\
\sum\limits_{\set{i\in\M\setminus\set{j}:\q{ji}{} >0}}\frac{\del[ji]{\M}\q{ji}{}}{\p{i}{\set{i,j}}}-\sum\limits_{\set{i\in\M\setminus\set{j}:\q{ji}{} =0, \q[c]{ji}{i,j} =0}}\frac{\del[ij]{\M}\q{ij}{}}{\p{j}{\set{i,j}}}=0,
	\end{gathered}
	\end{equation*}
	where we obtain the last equation by applying Lemma~\ref{lem1} on all pairs $(i,j)\in \G$ for which $\q{ji}{}\neq 0$ and use the fact that Assumption~\ref{a-TR-IIA} and $\q{ji}{}= 0$ imply that either $\q{ij}{}= 0$ or $\q[c]{ji}{i,j}= 0$. Note that applying Lemma~\ref{lem1} does not lead to a division by zero because if $\q{ji}{}> 0$, then $\q[c]{ji}{i,j}>0$ and by rationalizability, $\p{i}{\set{i,j}}>0$.
\end{proof}

%% file: input/text/app-3-proof-T1.tex
\subsection{Proof of Theorem~\ref{T1}}
\label{proof:T1}
%%%%

%%%% (a)-(b)

%%%%
\subsubsection{Equivalence between (i) and (ii)}
The necessity part of the argument follows directly from Proposition~\ref{prop:comparability1}. 
In order to show the sufficiency part, fix a menu $\M$ and let $\bp[\M]$ be  rationalizable by an \msc[] $\langle \mat[]{\M},\init[]{}\rangle$ and all its \comp{}s. We will show that $\mat{\M}$ satisfies the Kolmogorov criterion given in Equality~\eqref{kolmogorov}. We consider the \comp{} $R(\M)$ with positive entries except between the pair $i,j\in\M$ such that $r_{ij}(\M)=r_{ji}(\M)=0$. The rationalizability of~$\bp[\M]$ by $\mat{\M}$ and $R(\M)\odot \mat{\M}$ implies that ${\lra[]{\init[]{}}{\mat{\M}}(I-\mat[]{\M}) = \boldsymbol{0}}$ and ${\lra[]{\init[]{}}{\mat{\M}}(I-R(\M)\odot\mat[]{\M}) = \boldsymbol{0}}$. Consider the row of these linear systems corresponding to alternative $j$
	\begin{equation*}
	\begin{gathered}
	\sum_{k\neq j}\ro[]{}{k}{\M}\q{kj}{}-\ro[]{}{j}{\M}\sum_{k\neq j}\q{jk}{}=0,\\
	\sum_{k\neq j}\ro[]{}{k}{\M}r_{kj}(\M)\q{kj}{}-\ro[]{}{j}{\M}\sum_{k\neq j}r_{jk}(\M)\q{jk}{}=0.\\
	\end{gathered}
	\end{equation*}

	Using that $r_{kl}(\M)=c$ except for the pair $i,j$ we can express latter equation:
	\begin{equation*}
%	\begin{gathered}
%	\sum_{k\neq j,i}(\p{k}{\M}\q{kj}{}-\p{j}{\M}\q{jk}{})+\p{i}{\M}\q{ij}{}-%\p{j}{\M}\q{ji}{}=0,\\
	c\sum_{k\neq j,i}\ro[]{}{k}{\M}\q{kj}{}-\ro[]{}{j}{\M}\q{jk}{}=0.
%	\end{gathered}
	\end{equation*}
Thus, detailed balance given in Equation~\eqref{balance} is satisfied for the pair $i,j$ in the non-restricted model $\mat{\M}$. We can prove analogously that detailed balance holds on all pairs such that restricting their comparability does not lead to the formation of new communicating classes. Note that detailed balance holds trivially for all pairs which are not directly comparable in $\mat{\M}$. Thus, all cycles $\cya[\M']$ composed of these pairs satisfy the Kolmogorov criterion. 

Finally, we consider those pairs that belong to the same communicating class according to $\mat{\M}$, but belong to different classes when the transition between the two is prohibited. Let $\M'$ and $\M''$ be two disjoint subsets of $\M$ with $i\in\M'$ and $j\in\M''$ such that $\q{kl}{}=0$ for all $k\in\M'$ and $l\in\M''$ except $\q{ij}{}>0$. It easily seen that the Kolmogorov criterion is satisfied for all cycles with $(i,j)\in\cya[\M'\cup\M'']$. Therefore, the Kolmogorov criterion holds on all cycles of alternatives and the Markov process given by $\mat{\M}$ is reversible.

%%%%

%%%% (b)-(c)

%%%%
\subsubsection{Equivalence between (i) and (iii)}
\label{proof-T1-part2}
Consider a stochastic choice function $\bp[\Nset]$ and a menu $\M\in\Nset$. The condition $\cyrset[]{\M} = \cyrset[0]{\M}$ is trivially satisfied for $|\M|=2$, as is reversibility by all generating models. We assume $|\M|\geq 3$ for the rest of the proof. Let $\langle \mat[]{\M},\init[]{}\rangle$ be a rationalizing \msc[], \ie $\lra[]{\init[]{}}{\mat[]{\M}}=\bp[\M]$. Note that at least one rationalizing model always exists: letting \q{ij}{}= 0 for all $i,j\in\M$ and $\init[]{}=\bp[\M]$ generates the stochastic choice function $\bp[\M]$. \\

\textbf{(iii) $\implies$ (i):} We prove the statement by contraposition, \ie we show that if an \msc[] violates reversibility on~$M$, then $\bp[\M]$ is such that $\cyrset[+]{\M}\neq \emptyset$.

The proof is structured in three steps. First, we show that a non-reversible \msc[] on~$\M$ generates stochastic choice functions with $\del[kl]{\M}\neq 0$ for those pairs $k,l\in\M$ that violate the detailed balance conditions given in~\eqref{balance}. In the second step, we apply Lemma~\ref{lem4} from Appendix~\ref{sec:aux} which gives necessary and sufficient conditions for rationalizability. If the Markov chain is not reversible on~$\M$, the system has a positive solution. Finally, we invoke Gordan's theorem, which gives a necessary and sufficient condition for the existence of positive solutions to a homogeneous linear system and show the relationship of the condition to the property $\cyrset[]{\M} = \cyrset[0]{\M}$. 

\textbf{Step 1}: 
If an \msc[] is not reversible on~$\M$, then detailed balance is violated for at least one pair of alternatives: $ \exists k,l\in\M$ such that 
\begin{equation}
\q{kl}{}\p{k}{\M}>\q{lk}{}\p{l}{\M}.
\label{eq:violation-rev}
\end{equation}
Thus, $\q{kl}{}>0$ and $\p{k}{\M}>0$.
We consider the following two cases: (1) $\q[c]{lk}{k,l}=0$ and (2) $\q[c]{lk}{k,l}>0$. In the former, Assumption~\ref{a-pc} implies that $\q[c]{kl}{k,l}>0$ and thus, $\p{k}{\set{k,l}}=0$. Since $\p{k}{\M}>0$, we have that 
$\p{l}{\set{k,l}}\p{k}{\M}>\p{k}{\set{k,l}}\p{l}{\M}=0,$
and $\del[kl]{\M}>0$.
In the second case with $\q[c]{lk}{k,l}>0$, Assumption~\ref{a-TR-IIA} together with $\q{kl}{}>0$ imply that $\q[c]{kl}{k,l}>0$ and $\q{lk}{}>0$. The generated stochastic choice function is such that $\p{l}{\set{k,l}}\in (0,1)$. It follows from Inequality~\eqref{eq:violation-rev} and Assumption~\ref{a-TR-IIA} that
\begin{equation*}
\frac{\p{l}{\set{k,l}}}{\p{k}{\set{k,l}}}=\frac{\q[c]{kl}{k,l}}{\q[c]{lk}{k,l}}=\frac{\q{kl}{}}{\q{lk}{}}>\frac{\p{l}{\M}}{\p{k}{\M}}.
\end{equation*}
Thus, we have shown that whenever the detailed balance condition is violated for a pair of alternatives $k,l$ on $\M$, $\del[kl]{\M}\neq 0$.

\textbf{Step 2}: Recall from Appendix~\ref{proof:auxiliary} that $\G=\set{(i,j)\in\M^2:\del[ij]{\M}>0} = \set{g_b^{\M}}_{b\in\set{1,\dots,|\G|}}$. Our previous result ensures that $\G\neq\emptyset$. We apply Lemma~\ref{lem4}, which states that $\bp[\M]$ is such that
	\begin{equation}
	\mathcal{D}(\M)\sav{}=\boldsymbol{0}.
	\label{eq:diff-systemT1-1}
	\end{equation}
%	
%If the \msc[] \mat{\M} is reversible, then for each pair $i,j$ we either have $\dif{ij}{}=0$ or/and $\q{ij}{}=\q{ji}{}=0$. Therefore, reversible models correspond to the solution $\sav{}=\boldsymbol{0}$ of the system~\eqref{eq:diff-system}. Since this is a trivial solution, reversible \msc[] can generate all stochastic choice functions. On the other hand, 

\textbf{Step 3}: As already shown, whenever the detailed balance condition is violated for a pair of alternatives $k,l$ on $\M$ (and thus $\q{kl}{}>0$ and/or $\q{lk}{}>0$), $\del[kl]{\M}\neq 0$. This means that $(k,l)\in\G$ and $\sav{}\geq \boldsymbol{0}$. It follows directly from Gordan's theorem\footnote{See for example Theorem 15.1(2) in \citet{Woerdeman2015}.} that there exists a strictly positive solution ${\sav{}\geq \boldsymbol{0}}$ to the linear system~\eqref{eq:diff-systemT1-1} if and only if there does not exist a vector $\boldsymbol{z}=(z_1,\dots, z_{|\M|})\in\mathbb{R}^{|M|}$ such that
\begin{equation}
\label{eq:sys-strongly1}
\boldsymbol{z}\mathcal{D}(\M)\gg\boldsymbol{0}.
\end{equation}
Let the menu be denoted as an indexed set $\M=\set{m_a:a\in\set{1,\dots,|\M|}}$ and the elements of $\boldsymbol{z}$ be such that $z_a=v_k$ whenever $m_a=k$. The linear system~\eqref{eq:sys-strongly1} is equivalent to
\begin{equation*}
z_{a}\mathrm{d}_{a,b}(\M) + z_{c}\mathrm{d}_{c,b}(\M)>0, \quad \forall (k,l) \in \G\text{ with }m_a=k, m_c=l, g_{b}^M=(k,l),
\end{equation*}
which is in turn equal to
\begin{equation}
\label{eq:v-diff}
(v_k - v_l)\frac{\del[kl]{\M}}{\p{l}{\set{k,l}}}>0, \quad \forall (k,l) \in \G.
\end{equation}
Hence, in order for Inequality~\eqref{eq:sys-strongly1} to hold, the vector $\boldsymbol{z}$ should be such that for all pairs $(k,l) \in\G$, we have $v_k>v_l$. Since the Markov chain is non-reversible on~$\M$ and $\sav{}\geq \boldsymbol{0}$, there does not exist a vector $\boldsymbol{z}$ satisfying \eqref{eq:v-diff}. This is the case when there exists some cycle $\cya[\M']\in\cyrset[+]{\M}$ and $\cya[\M']\subseteq\G$, because then we cannot assign a value to each element of $\M'$ such that $v_k>v_l$ for all $(k,l)\in \cya[\M']$.

\textbf{(i) $\implies$ (iii)}: We again prove the statement by contraposition, \ie we show that if \bp[\M] is such that $\cyrset[]{\M}\neq\cyrset[0]{\M}$, then there exists a non-reversible rationalizing \msc[]. The proof is structured analogously to the sufficiency part.

\textbf{Step 1}: The assumption that $\cyrset[]{\M}\neq\cyrset[0]{\M}$ implies that there exists a positive sign-consistent cycle $\cya[\M']\in\cyrset[+]{\M}$ and $\cya[\M']\subseteq\G$. Since $\G\neq\emptyset$, it follows from Lemma~\ref{lem4} that the transition probabilities of all rationalizing models satisfy $\mathcal{D}(\M)\sav{}=\boldsymbol{0}$.

\textbf{Step 2}: Gordan's theorem implies that there exists a strictly positive solution ${\sav{}\geq \boldsymbol{0}}$ if and only if there is no vector $\boldsymbol{z}\in\mathbb{R}^{|M|}$ for which
\begin{equation}
\label{eq:sys-strongly1T1-2}
\boldsymbol{z}\mathcal{D}(\M)\gg\boldsymbol{0}.
\end{equation}
Letting $z_a=v_k$ whenever $m_a=k$, the system becomes equivalent to
\begin{equation}
\label{eq:v-diffT1-2}
(v_k - v_l)\frac{\del[kl]{\M}}{\p{l}{\set{k,l}}}>0, \quad \forall (k,l) \in \G.
\end{equation}	
Since it holds for all $(k,l)\in\cya[\M']$ that $\del[kl]{\M}>0$, we cannot assign a numerical value to each element of $\M'$ such that $v_k>v_l$ for all $(k,l)\in \cya[\M']$. Therefore, there does not exist a vector $\boldsymbol{z}\in\mathbb{R}^{|M|}$ that satisfies the system of inequalities~\eqref{eq:sys-strongly1T1-2} and there exists a rationalizing model is such that ${\sav{}\geq \boldsymbol{0}}$.

\textbf{Step 3}: We can now use that there is a solution ${\sav{}\geq \boldsymbol{0}}$ to construct a transition probability matrix $\mat[]{\M}$. In particular, we define
\begin{align}
\qp{ij}{}=\begin{cases} 
\gamma_{b}(\M)& \text{ if }g_{b}^M=(i,j)\in\G,\\
\gamma_{b}(\M)\frac{\p{j}{\set{i,j}}}{\p{i}{\set{i,j}}}& \text{ if }g_{b}^M=(j,i)\in\G\\
\p{j}{\set{i,j}}& \text{ if }\del[ij]{\M}=0.
\end{cases}
\label{qij}
\end{align}
Note that $(j,i) \in\G$ implies that $\p{i}{\set{i,j}}>0$ since $\del[ji]{\M}>0$.
Next, we let $\q{ij}{}=\kappa\qp{ij}{}$ for all $i\neq j$, where $\kappa>0$ is such that $\kappa\sum_{j\neq i}\qp{ij}{}< 1$ for all $i\in\M$. Finally, we let $\q{ii}{}=1-\sum_{j\neq i}\q{ij}{}$. The constructed transition probability matrix satisfies Assumptions~\ref{a-consideration} --~\ref{a-TR-IIA} and rationalizes \bp[\M].

Finally, we will show that there exists a pair  of alternatives that violates detailed balance. From the above definition of the model and the fact that ${\sav{}\geq \boldsymbol{0}}$, there must exist a pair $(k,l)\in\G$, \ie $\del[kl]{\M}=\p{k}{\M}\p{l}{\set{k,l}}-\p{l}{\M}\p{k}{\set{k,l}}>0$, for which $\q{kl}{\M}>0$. Therefore, $\p{l}{\set{k,l}}>0$ and $\p{k}{\M}>0$.
We consider two cases: (1) $\p{k}{\set{k,l}}=0$ and (2) $\p{k}{\set{k,l}}>0$. In the first case, it follows from rationalizability and Assumptions~\ref{a-consideration} --~\ref{a-TR-IIA} that $\q{lk}{\M}=0$. Since $\q{kl}{\M}\p{k}{\M}>0$ and $\q{lk}{\M}=0$, detailed balance is violated on the pair $k,l$ and the \msc[] is non-reversible. In the second case with $\p{k}{\set{k,l}}>0$, we have $\q{lk}{}=\q{kl}{}\frac{\p{k}{\set{k,l}}}{\p{l}{\set{k,l}}}$. Thus, $\del[kl]{\M}>0$ implies $\frac{\p{l}{\set{k,l}}}{\p{k}{\set{k,l}}}=\frac{\q{kl}{}}{\q{lk}{}}>\frac{\p{l}{\M}}{\p{k}{\M}}$, which is a violation of detailed balance.

%% file: input/text/app-3-proof-comparability1.tex
\subsection{Proof of Proposition~\ref{prop:comparability1}}
\label{proof:comparability1}
Consider a reversible and irreducible \msc[] $\langle \mat[]{\M},\init[]{}\rangle$ and its generated stochastic choice function $\brho[l]{}$ for a menu $\M\in\Mset$. Let $R(\M)$ be an arbitrary \comp[] of $\langle \mat[]{\M},\init[]{}\rangle$. We will show that $\langle R(\M)\odot\mat{\M},\init[]{}\rangle$ is reversible and irreducible and $\brho[l]{}$ is its stationary distribution.

Since a \comp preserves the communicating classes, $R(\M)\odot\mat{\M}$ is irreducible and has a unique stationary distribution. The reversibility of $\mat[]{\M}$ implies that the Kolmogorov's criterion given in Equation~\eqref{kolmogorov} is satisfied. The symmetry of $R(\M)$ implies that the Kolmogorov condition is also satisfied by $R(\M)\odot\mat{\M}$ and is reversible. 

To prove rationalizability, we need to show ${\lra[]{\init[]{}}{\mat{\M}}(I-R(\M)\odot\mat[]{\M}) = \boldsymbol{0}}$. We can express the equation for all $j\in\M$ as 
	\begin{align*}
	\begin{gathered}
	\sum_{i\neq j}\ro[]{}{i}{\M}r_{ij}(\M)\q{ij}{}-\ro[]{}{j}{\M}\sum_{i\neq j}r_{ji}(\M)\q{ji}{}=0,\\
	\sum_{i\neq j}r_{ij}(\M)\left(\ro[]{}{i}{\M}\q{ij}{}-\ro[]{}{j}{\M}\q{ji}{}\right)=0,
	\end{gathered}
	\end{align*}
where the last step results from the symmetry of $R(\M)$. Since $\mat{\M}$ must satisfy the detailed balance condition given in Equation~\eqref{balance} on all pairs, one can easily see that the equality holds. 
	
Finally, consider the case when the initial \msc[] is not irreducible. Since \comp{}s do not affect the initial distribution and preserve the communicating classes, our proof extends trivially to the case with reducible \msc[]s.

%% file: input/text/app-3-proof-T2.tex
\subsection{Proof of Theorem~\ref{T2}}
\label{proof:T2}
Consider a stochastic choice function $\bp[\Nset]$ and a menu $\M\in\Nset$. Since both parts of the theorem are trivially satisfied for $|\M|=2$, we assume $|\M|\geq 3$ for the rest of the proof.

\textbf{(i) $\implies$ (ii):} We consider a pairwise comparable \msc[] on $\M$ with generated stochastic choice function $\lra[]{\init[]{}}{\mat[]{\M}}=\bp[\M]$. If it holds for all pairs $i,j\in\M$ that $\del[ij]{\M}=0$, the condition in (ii) is trivially satisfied. We assume hereafter that there is at least one pair $i,j\in\M$ with $\del[ij]{\M}>0$.

We apply Lemma~\ref{lem4} from Appendix~\ref{sec:aux}. Since $\del[ij]{\M}>0$, $(i,j)\in\G$ by definition. The assumption that the model is pairwise comparable, \ie there is no $k,l\in\M$ for which $\q{kl}{}=\q{lk}{}=0$, ensures that $\sav{}\gg\boldsymbol{0}$. We use Stiemke's lemma,\footnote{See for example Theorem 15.1(1) in \citet{Woerdeman2015}.} which states that there exists a strongly positive solution $\sav{}$ to the linear system in~\eqref{eq:diff-system} if and only if there does not exist a vector $\boldsymbol{z}\in\mathbb{R}^{|M|}$ for which
\begin{equation}
\label{eq:sys-strongly-ii}
\boldsymbol{z}\mathcal{D}(\M)\geq\boldsymbol{0}.
\end{equation}
Analogously to the simplification presented in the proof of Theorem~\ref{T1} given in Inequality~\eqref{eq:v-diff}, this system can be written as
\begin{equation*}
\label{eq:v-diff1}
(v_k - v_l)\frac{\del[kl]{\M}}{\p{l}{\set{k,l}}}\geq 0, \forall (k,l) \in \G \text{ and } (v_k - v_l)\frac{\del[kl]{\M}}{\p{l}{\set{k,l}}}> 0 \text{ for at least one }(k,l) \in \G.
\end{equation*}
Since $\sav{}\gg\boldsymbol{0}$, such vector $\boldsymbol{z}$ does not exist. Note that it must hold for $k,l\in\M$ that belong to a cycle in $\cyrset[+]{\M}$ that $v_k=v_l$. Since a vector~$\boldsymbol{z}$ does not exist, we should not be able to find a pair $(k,l)\in \G$ such that $v_k>v_l$. This is the case when $\bp[\M]$ is such that all pairs $(k,l)\in \G$ belong to some sign-consistent cycle in \cyrset[+]{\M}, which coincides with condition (ii).

Finally, if the \msc[] is fully comparable, the positivity assumption of the generated stochastic choice function follows directly from the Perron–Frobenius theorem.

\textbf{(ii) $\implies$ (i):}
Consider a stochastic choice function such that for all pairs $i,j\in\M$ with $\del[ij]{\M}\neq 0$, $\exists \cya[\M']$ with $(i,j)\in\cya[\M']$ and $\cya[\M']\in\cyrset[]{\M}$ for $\M'\in\Mset$. We will show that there exists a rationalizing \msc[] with $\mat[]{\M}$ that is pairwise comparable. Note that all stochastic choice functions are rationalizable by an \msc[] (see the Appendix~\ref{proof-T1-part2}). 

Assume first that $\del[ij]{\M}=0$ for all pairs $i,j\in\M$. We apply Lemma~\ref{lem3} from Appendix~\ref{sec:aux} and see that Equation~\eqref{eq:diff} is trivially satisfied for all $j\in\M$. A pairwise comparable \msc[] that rationalizes the stochastic choice function is, for example, one where $\q{ij}{}=\kappa\p{j}{\set{i,j}}$, where $\kappa>0$ is such that $\mat[]{\M}$ satisfies Assumptions~\ref{a-consideration} --~\ref{a-TR-IIA}.

We assume hereafter that there is at least one pair $i,j\in\M$ with $\del[ij]{\M}>0$ and thus $(i,j)\in\G$. It follows from Lemma~\ref{lem4} that the transition probabilities satisfy the system 
\begin{equation}
	\label{eq:diff-systemT2-2}
	\mathcal{D}(\M)\sav{}=\boldsymbol{0}.
	\end{equation} 
Following Stiemke's lemma, there exists a strongly positive solution ${\sav{}\gg \boldsymbol{0}}$ to the system~\eqref{eq:diff-systemT2-2} if and only if there does not exist a vector $\boldsymbol{z}\in\mathbb{R}^{|M|}$ for which
\begin{equation}
\label{eq:sys-strongly1T2-2}
\boldsymbol{z}\mathcal{D}(\M)\geq\boldsymbol{0}.
\end{equation}
Analogously to the proof of Theorem~\ref{T1}, the above linear system is equivalent to
\begin{equation*}
(v_k - v_l)\frac{\del[kl]{\M}}{\p{l}{\set{k,l}}}\geq 0, \forall (k,l) \in \G \text{ and } (v_k - v_l)\frac{\del[kl]{\M}}{\p{l}{\set{k,l}}}> 0 \text{ for at least one }(k,l) \in \G.
\end{equation*}	
Since it holds for all alternatives $(k,l)\in\G$ that $\exists \cya[\M']$ with $(k,l)\in\cya[\M']$ and $\cya[\M']\in\cyrset[]{\M}$ for some $\M'\in\Mset$, we have that $v_k=v_l$ and there is no pair $(k,l)\in\G$ for which $(v_k - v_l)\frac{\del[kl]{\M}}{\p{l}{\set{k,l}}}> 0$. Thus, there exists a rationalizing model with ${\sav{}\gg \boldsymbol{0}}$. 

We can use the solution to the system~\eqref{eq:diff-systemT2-2} to construct a transition probability matrix $\mat[]{\M}$ as follows:
\begin{align}
\qp{ij}{}=\begin{cases} 
\gamma_{b}(\M)& \text{ if }g_{b}^M=(i,j)\in\G,\\
\gamma_{b}(\M)\frac{\p{j}{\set{i,j}}}{\p{i}{\set{i,j}}}& \text{ if }g_{b}^M=(j,i)\in\G\\
\p{j}{\set{i,j}}& \text{ if }\del[ij]{\M}=0.
\end{cases}
\label{qij-T2}
\end{align}
Note that $(j,i) \in\G$ implies that $\p{i}{\set{i,j}}>0$ since $\del[ji]{\M}>0$.
Next, we let $\q{ij}{}=\kappa\qp{ij}{}$ for all $i\neq j$, where $\kappa>0$ is such that $\kappa\sum_{j\neq i}\qp{ij}{}< 1$ for all $i\in\M$. Finally, we let $\q{ii}{}=1-\sum_{j\neq i}\q{ij}{}$. This model satisfies Assumptions~\ref{a-consideration} --~\ref{a-TR-IIA}, rationalizes \bp[\M], and is pairwise comparable.

Finally, we show that assuming positivity in addition, ensures the rationalizability with a fully comparable model. Since ${\sav{}\gg \boldsymbol{0}}$, it holds for all $(i,j)\in\G$ that $\q{ij}{\M}>0$. It is easy to see that the positivity of the stochastic choice function ensures that all other transition probabilities are positive as well, hence the constructed model is fully comparable.

%% file: input/text/app-4-proof-edge.tex
\subsection{Proof of Proposition~\ref{prop:no-edges}}
\label{proof:no-edges}
Since $\bp[\M]$ is not \btc over the pair $i,j$, it holds $\del[ji]{\M}\neq 0$ and $\nexists \cya[\M']$ with $(i,j)\in\cya[\M']$ and $\cya[\M']\in\cyrset[+/-]{\M}$ for $\M'\in\Mset$. We assume by contradiction that there is a rationalizing model $\langle \mat[]{\M},\init[]{}\rangle$ with $\q{ij}{}>0$. Since $\bp[\M]$ violates Theorem~\ref{T2}(ii), there has to be at least one pair $k,l\in\M$ for which $\q{kl}{}=0$. Let us denote the set of such pairs with $S(\mat{\M})$. 

We construct a stochastic choice function $\bp[\Nset]^*$ such that $\bp[\M]^*=\bp[\M]$ and for all $k,l\in S(\mat{\M})$ we adjust the choice probability from the binary sets so that $\delta_{kl} (\bp[\M]^*)= 0$. Note that $\bp[\M]^*\neq \bp[\M]$ since following the proof of Theorem~\ref{T2}, there is at least one pair with $\del[kl]{\M}\neq 0$ and $\q{kl}{}=0$. Note further that $\bp[\M]^*$ is not \btc over the pair $i,j$, since according to the described procedure, $\mathcal{C}_{+/-}^{S}(\bp[\M]^*)\subseteq\cyrset[+/-]{\M}$ and $\delta_{ij} (\bp[\M]^*)\neq 0$.

We now construct an \msc[] with $\mat{\M}^*$ such that $\q{kl}{}^*=\q{kl}{}$ for all $k,l\not\in S(\mat{\M})$ and let $\q{kl}{}^*\in(0,1)$ and $\q{lk}{}^*=\q{kl}{}^*\frac{\p{k}{\set{k,l}}^*}{\p{l}{\set{k,l}}^*}$ for all $k,l\in S(\mat{\M})$. Observe that this model is pairwise comparable.

Applying Lemma~\ref{lem4} shows that the constructed \msc[] $\mat{\M}^*$ rationalizes $\bp[\M]^*$. This is the case because $\delta_{kl} (\bp[\M]^*)\q{kl}{}^* = \del[kl]{\M}\q{kl}{}$ for all $k,l \in \M$. Theorem~\ref{T2} implies that $\bp[\M]^*$ should be \btc over all pairs of alternatives, which is a contradiction. Therefore, all rationalizing models of $\bp[\Nset]$ should be such that $\q{ij}{}=0$. 

%% file: input/text/app-3-proof-luce.tex
\subsection{Proof of Theorem~\ref{prop:luce}}
\label{proof:luce}
Consider a stochastic choice function $\bp[\Nset]$ and a rationalizing \msc[] on $\M\in\Nset$, \ie $\lra[]{\init[]{}}{\mat[]{\M}}=\bp[\M]$. 
\subsubsection{Equivalence between (i) and (ii)}
%%%
\textbf{(i) $\implies$ (ii):}  Let $\mat{\M}$ be a reversible and fully comparable \msc[]. We fix an arbitrary pair $i,j\in\M$, for which will show that the stochastic choice function satisfies positivity and IIA. 

Since the model is fully comparable, the stationary distributions are such that $\ro[]{}{i}{\M}>0$ for all $i\in\M$ and all $\M\in\Mset$, as there are no transient states. Hence, the generated stochastic choice function is positive. Reversibility means that the detailed balance condition is satisfied for all pairs and menus, hence 
\[\frac{\ro[]{}{i}{\M}}{\ro[]{}{j}{\M}}=\frac{\q{ji}{}}{\q{ij}{}} \text{  and  }\frac{\ro[c]{}{i}{i,j}}{\ro[c]{}{j}{i,j}}=\frac{\q[c]{ji}{i,j}}{\q[c]{ij}{i,j}}.\]
It follows from Assumption~\ref{a-TR-IIA} that all fractions in the above equations are equal and hence IIA is satisfied on the pair $i,j$ and thus, on all pairs and menus.

\textbf{(ii) $\implies$ (i):}  Let $\bp[\Nset]$ be a positive stochastic choice function that satisfies IIA, thus $\del[kl]{\M}=0$ for all $k,l\in\M$. Lemma~\ref{lem3} in Appendix~\ref{proof:auxiliary} implies that there exists a fully comparable rationalizing \msc[]. We will show that the detailed balance condition given in Equation~\eqref{balance} holds on an arbitrary pair $i,j\in\M$.

Using the positivity and IIA properties of the choice function, we obtain
\begin{equation*}
\frac{\p{j}{\M}}{\p{i}{\M}}=\frac{\p{j}{\set{i,j}}}{\p{i}{\set{i,j}}}=\frac{\q[c]{ji}{i,j}}{\q[c]{ij}{i,j}}.
\end{equation*}
Finally, Assumption~\ref{a-TR-IIA} and the above equation imply detailed balance. Thus, the condition is satisfied on all pairs and menus and the rationalizing model is reversible.

%%%%
\subsubsection{Equivalence between (i) and (iii)}%%%

\textbf{(iii) $\implies$ (i):}  If there exists a utility function $u:\all\rightarrow \mathbb{R}_{++}$ such that the ratios of transition probabilities satisfy $\frac{\q{ij}{}}{\q{ji}{}}=\frac{u(j)}{u(i)}$ for all $i,j\in\M$, reversibility of the \msc[] is trivial as it follows directly from Kolmogorov's criterion given in Equation~\eqref{kolmogorov}.

\textbf{(i) $\implies$ (iii):} Let \mat{\M} be a reversible and fully comparable \msc[]. Hence, the detailed balance condition~\eqref{balance} is satisfied for all pairs of alternatives. As we showed previously, the generated stochastic choice function is positive and IIA. It follows from \citet{Luce1959} that there exists an increasing function ${u:\all\rightarrow \mathbb{R}_{++}}$ such that for all $\M \in \Mset$ and $i\in \M$ such that
\begin{equation}
\p{i}{\M}=\frac{\val{i}}{\sum_{j\in \M} \val{j}}
\label{utility}
\end{equation}
Plugging in Equation~\eqref{utility} into the detailed balance condition, we obtain that for all $i,j \in\M$ and all $\M\in\Mset$
\begin{equation*}
\begin{gathered}
\q{ji}{}\frac{\val{j}}{\sum_{k\in \M} \val{k}}=\frac{\val{i}}{\sum_{k\in \M} \val{k}}\q{ij}{}
\end{gathered}
\end{equation*}
and the result follows.

%% file: input/text/app-4-proof-T3.tex
\subsection{Proof of Theorem~\ref{T3}}
\label{proof:T3}
The \msc[] is finite and aperiodic by definition. If it is irreducible in addition, the Markov chain is ergodic. Thus, the equivalence of (i) and (ii) follows trivially from Proposition~\ref{prop:limit}, since the stationary distribution of an ergodic Markov process does not depend on the initial distribution.
\subsubsection{Proof of equivalence between statements (i) and (iii)}%%%
Consider a stochastic choice function $\bp[\Nset]$ and a rationalizing \msc[] $\langle \mat[]{\M},\init[]{}\rangle$ on $\M\in\Nset$. Note that a rationalizing model always exists (see Appendix~\ref{proof-T1-part2}).

\textbf{(i) $\implies$ (iii):} Let $\langle \mat[]{\M},\init[]{}\rangle$ be an irreducible model, \ie every state can be reached from every other state. In other words, for a given pair $(i,j)\in\M$, there exists some $\cya[\M']$ for which $\q{kl}{}>0$ for all $(k,l)\in\cya[\M']\setminus (j,i)$. Note that if we take the union of all such cycles corresponding to each of the pairs, we obtain a cycle $\cya[\M]$ for which $\q{kl}{}>0$ for all $(k,l)\in\cya[\M]$. We need to show that the generated \bp[\M] is \btc over all $(k,l)\in\cya[\M]$, \ie if $\del[kl]{\M}\neq 0$, then $\exists \cya[\M'']$ with $(k,l)\in\cya[\M'']$ and $\cya[\M'']\in\cyrset[]{\M}$ for some $\M''\in\Mset$.

Suppose by contradiction that there is a pair $(k,l)\in\cya[\M]$, such that $\del[kl]{\M}\neq 0$, but there is no sign-consistent cycle $\cya[\M'']$ with $(k,l)\in\cya[\M'']$ and $\cya[\M'']\in\cyrset[]{\M}$. Proposition~\ref{prop:no-edges} implies that all rationalizable models of $\bp[\M]$ are such that $\q{ij}{}=\q{ji}{}=0$, which is a contradiction. Thus, $\bp[\M]$ needs to be \btc over the pairs of the cycle $\cya[\M]$.
%If the detailed balance condition holds between some consecutive elements of the sequence $i,j$, therefore ${\ro[]{}{i}{\M}\q{ij}{}=\ro[]{}{j}{\M}\q{ji}{}}$. TR-IIA then implies that $\ro[]{}{i}{\M}\ro[c]{}{j}{i,j}=\ro[]{}{j}{\M}\ro[c]{}{i}{\M}$ and hence $\dif{ij}{}=0$. 
%
%Consider now the consecutive pairs in the sequence for which detailed balance is violated. Analogously, the generated stochastic choice function is such that $\dif{ij}{}\neq 0$ for all such pairs $i,j$. Assume by contradiction that $i,j$ do not belong to any cycle of \spo. Let $\brho{}=\bp[\M]$.  Hence, it holds for all consecutive pairs of the sequence for which $\dif{ij}{}\neq 0$ that they belong to some cycle of $\spo$.

\textbf{(iii) $\implies$ (i):} Let \bp[\M] be \btc over all pairs of alternative of some cycle $\cya[\M]$. We will show that there is a model $\langle \mat[]{\M},\init[]{}\rangle$ such that $\q{kl}{}>0$ for all $(k,l)\in\cya[\M]$, as it implies that every state can be reached from every other state and the Markov chain is irreducible. 

We construct a stochastic choice function $\bp[\Nset]^*$ such that $\bp[\Nset]^*=\bp[\Nset]$ except for the pairs $k,l \in \M$ for which \bp[\M] is not \btc, we adjust the choice probability from the binary sets so that $\delta_{kl} (\bp[\M]^*)= 0$. Thus, $\bp[\M]^*$ is \btc over all pairs $k,l \in \M$. Following Theorem~\ref{T2}, $\bp[\M]^*$ is rationalizable with a pairwise comparable \msc[] $\langle \mat[*]{\M},\init[]{}\rangle$.

We let $\mat{\M}$ be such that $\q{ij}{}=\q{ij}{}^*$ for all pairs $i,j\in\M$ over which \bp[\M] is \btc{} and $\q{ij}{}=\q{ji}{}=0$ for the remaining pairs in \M. It follows from the rationalizability of $\bp[\Nset]^*$ and Lemma~\ref{lem4} in Appendix~\ref{sec:aux}, that $\bp[\M]$ is rationalizable by $\langle \mat[]{\M},\init[]{}\rangle$ since $\delta_{kl} (\bp[\M]^*)\q{kl}{}^* = \del[kl]{\M}\q{kl}{}$ for all $k,l \in \M$. By definition, $\q{ij}{}=\q{ij}{}^*>0$ for all pairs $i,j\in\cya[\M]$ since \bp[\M] is \btc over those pairs and the model is irreducible.

% We will first show that the rationalizing model is such that for all $\M\in\Mset$ and $i,j\in\M$ it holds that either ${\q{ij}{}>0},{\q{ji}{}>0}$ or $\q{ij}{}=\q{ji}{}=0$. Then, we will show that there exists an irreducible \msc[], in particular that for all consecutive elements in the above sequence the transition probability is strictly positive. Since \msc[]s are aperiodic by definition, proving irreducibly implies that the \msc[] is ergodic.

%Suppose that there is a rationalizing \msc[] with $\q{ij}{}>0$ and $\q{ji}{}=0$. Since $\bp[\Nset]$ is positive on binary choice sets, we must have $\q[c]{ij}{i,j}>0$ and $\q[c]{ji}{i,j}>0$. Under these assumptions, TR-IIA is violated and we have found a contradiction.

%% file: input/text/app-2-ex-decoy.tex
\subsection{Details on Example~\ref{ex:decoy}}
\label{proof:decoy}
We first show that $\frac{\lra[i]{}{\mat[]{\set{i,j,k}}}}{\lra[j]{}{\mat[]{\set{i,j,k}}}}>\frac{\lra[i]{}{\mat[]{\set{i,j}}}}{\lra[j]{}{\mat[]{\set{i,j}}}}$ whenever $\q[c]{ik}{i,k}=\q[c]{ik}{i,j,k}=0$.
The rationalizability of $\bp[\Nset]$ implies $\bp[\M](I-\mat[]{\M}) = 0$ for all $\M\in\Nset$. Therefore, it holds for all $j \in \M$ and $\M\in\Nset$ that
		\begin{equation*}
		\sum\limits_{\substack{i\neq j}}\lra[j]{}{\mat[]{\M}}\q{ji}{}-\sum\limits_{\substack{i\neq j}}\lra[i]{}{\mat[]{\M}}\q{ij}{}=0.
		\end{equation*}
		In particular, it holds $
\lra[j]{}{\mat[]{\all}}\q[2]{ji}{\all}+\lra[k]{}{\mat[]{\all}}\q[2]{ki}{\all}=\lra[i]{}{\mat[]{\all}}\q[2]{ij}{\all}.
$
After applying Assumption~\ref{a-TR-IIA} and rearranging, we obtain
			\begin{equation}
\frac{\lra[j]{}{\mat[]{\all}}}{\lra[j]{}{\set{i,j}}}+\lra[k]{}{\mat[]{\all}}\frac{\q[2]{ki}{\all}}{\q[2]{ji}{\all}}=\frac{\lra[i]{}{\mat[]{\all}}}{\lra[i]{}{\set{i,j}}}.
\label{eq-pi}
\end{equation} 
Since the model is irreducible, the induced choice function is positive (Theorem~\ref{T3}). Thus, all terms in the equation are positive and the first result follows. 

We now prove the second part of the statement that $\q[c]{ki}{i,j,k}>\q[c]{ji}{i,j,k}$ if and only if $\lra[i]{}{\mat[]{\set{i,j,k}}}>\lra[i]{}{\mat[]{\set{i,j}}}$. Starting from Equation~\eqref{eq-pi}, we use that all choice probabilities from the set $\all$ should sum up to 1:
			\begin{equation*}
\frac{\lra[j]{}{\mat[]{\all}}}{\lra[j]{}{\set{i,j}}}+(1-\lra[j]{}{\mat[]{\all}})\frac{\q[2]{ki}{\all}}{\q[2]{ji}{\all}}=\frac{\lra[i]{}{\mat[]{\all}}}{\lra[i]{}{\set{i,j}}}+\frac{\q[2]{ki}{\all}}{\q[2]{ji}{\all}}\lra[i]{}{\mat[]{\all}}.
\end{equation*}
Rearranging the equality and using that $\lra[i]{}{\mat{\all}}>\lra[i]{}{\mat{\set{i,j}}}$, we obtain
\begin{align*}
\begin{gathered}
\frac{\lra[j]{}{\mat[]{\all}}}{\lra[j]{}{\set{i,j}}}\q[2]{ji}{\all}+(1-\lra[j]{}{\mat[]{\all}})\q[2]{ki}{\all}>\q[2]{ji}{\all}+\q[2]{ki}{\all}\lra[i]{}{\set{i,j}},\\
(\q[2]{ji}{\all}-\q[2]{ki}{\all})(\lra[j]{}{\mat[]{\all}}-\lra[j]{}{\set{i,j}})>0.
\end{gathered}
\end{align*}
Note that $\lra[i]{}{\mat{\all}}>\lra[i]{}{\mat{\set{i,j}}}$ together with the fact that all choice probabilities from a given menu sum up to 1 implies that $\lra[j]{}{\set{i,j}}>\lra[j]{}{\mat[]{\all}}$ and the result follows.

%% file: input/text/app-3-ex-T1.tex
\subsection{Details on Example~\ref{ex:T1}}
\label{app-ex:T1}
Note that Assumption~\ref{a-TR-IIA} implies $\frac{\q{ji}{}}{\q{ij}{}}=\frac{\q[c]{ji}{i,j}}{\q[c]{ij}{i,j}}=\frac{\p[]{i}{\set{i,j}}}{\p[]{j}{\set{i,j}}}$. Since the stochastic choice function on binary sets is the same, all rationalizing \msc[]s in Examples~\ref{ex:T1}-\ref{ex:T3} are such that:
\begin{equation}
\label{Qex}
I-\mat[]{\M}= 
    \begin{pmatrix}
q_{ij}+q_{ik}+q_{il}	&-q_{ij}	&-q_{ik}	&-q_{il}\\
-q_{ij}	&q_{ij}+q_{jk}+q_{jl} 	&-q_{jk}	&-q_{jl}  \\
-q_{ik} 	&-\frac{3}{2}q_{jk}		&q_{ik}+\frac{3}{2}q_{jk}+q_{kl} 	&-q_{kl}  \\
-q_{il}	&-q_{jl}	&-\frac{2}{3}q_{kl}&q_{il}+q_{jl}+\frac{2}{3}q_{kl}
\end{pmatrix}.
\end{equation}

We show that no rationalizing \msc[] will have positive transition probabilities to and from state $k$. Using the definition of rationalizability, it has to holds that
\begin{align*}
\begin{gathered}
\p{i}{\M}q_{ik}+\p{j}{\M}q_{jk}-\p{k}{\M}(q_{ik}+\frac{3}{2}q_{jk}+q_{kl})+\p{l}{\M}\frac{2}{3}q_{kl}=0,\\
q_{ik}+q_{jk}-2(q_{ik}+\frac{3}{2}q_{jk}+q_{kl})+\frac{2}{3}q_{kl}=0,
\end{gathered}
\end{align*}
where the second equation results from plugging in the stochastic choice function and simplifying.
Since transition probabilities cannot be negative, the only solution is $\q{ik}{}=\q{jk}{}=\q{kl}{}=0$.

One can easily verify that $\bp[\M] = (0.2,0.2,0.4,0.2)$ is the stationary distribution of an \msc[] with the following transition probability matrix and $\init[k]{}=0.4$ 
\begin{equation*}
Q_2(\M)= 
    \begin{pmatrix}
0.8	&0.1	&0	&0.1\\
0.1	&0.8	&0	&0.1  \\
0 	&0		&1	&0  \\
0.1	&0.1	&0	&0.8
\end{pmatrix}.
\end{equation*}
Note that restricting the comparability between each one of the pairs results in another rationalizing \msc[] when the initial distribution is the same. For example, if we let $R(\M)$ be such that $r_{il}(\M)=0$ and all other off-diagonal elements equal to 1, $R(\M)\odot Q_2(\M)$ rationalizes $\bp[\M]$. This holds even if the \comp{} is weak, for example $r_{il}(\M)=0$, $r_{ij}(\M)=2$ and all other off-diagonal elements equal to 1.
%\begin{equation*}
%\bp[\M]R(\M)\odot Q_2(\M)=\bp[\M]    \begin{pmatrix}
%0.9	&0.1	&0	&0\\
%0.1	&0.8	&0	&0.1  \\
%0 	&0		&1	&0  \\
%0	&0.1	&0	&0.9
%\end{pmatrix}=\bp[\M].
%\end{equation*}
%Even if the comparability restriction is weak, \ie $r_{il}(\M)=0$, $r_{ij}(\M)=2$ and all other off-diagonal elements equal to 1, we obtain 
%
%\begin{equation*}
%\bp[\M]R(\M)\odot Q_2(\M)=\bp[\M]    \begin{pmatrix}
%0.8	&0.2&0	&0\\
%0.2&0.7	&0	&0.1  \\
%0 	&0		&1	&0  \\
%0	&0.1	&0	&0.9
%\end{pmatrix}=\bp[\M].
%\end{equation*}

%% file: input/text/app-3-ex-T2.tex
\subsection{Details on Example~\ref{ex:T2}}
\label{app-ex:T2}
As pointed out in Appendix~\ref{app-ex:T1}, the transition probability matrix should satisfy Equation~\eqref{Qex} and ${\brho[l]{}(I-\mat[]{\M}) = \boldsymbol{0}}$. One can verify that the below transition probability matrix satisfies these requirements for $\bp[\M] = (0.25,0.28,0.2,0.27)$:
\begin{equation*}
Q_3(\M)= 
    \begin{pmatrix}
0.7	&0.1	&0.1	&0.1\\
0.1	&0.72	&0.16	&0.02  \\
0.1	&0.24	&0.57	&0.09  \\
0.1	&0.02	&0.06	&0.82
\end{pmatrix}.
\end{equation*}

%We can also verify that for each pair there is some sign-consistent cycle such that for all pairs in the cycle Inequality~\eqref{eq:bc} holds. Consider the pair $i,j\in\M$ and the sign-consistent cycle $\set{(i,j),(j,k),(k,i)}\in\cyrset[-]{\M}$. Plugging in the stochastic choice function in Inequality~\eqref{eq:bc}, we obtain for the pair $(i,j)$:
%\begin{align*}
%\frac{\p{i}{\M}}{\p{j}{\M}}=\frac{0.25}{0.28}\in  \left(\frac{\p[]{k}{\set{j,k}}}{\p[]{j}{\set{j,k}}}\frac{\p[]{i}{\set{i,k}}}{\p[]{k}{\set{i,k}}},\frac{\p[]{i}{\set{i,j}}}{\p[]{j}{\set{i,j}}}\right)=\left(\frac{0.4}{0.6},1\right).
%\end{align*}
%We can verify that the inequality holds for the other two pairs in the cycle.

%% file: input/text/app-4-ex-T3.tex
\subsection{Details on Example~\ref{ex:T3}}
First, we show that there is no rationalizing model which assigns a positive transition probability between the pairs $(j,l)$ and $(k,l)$ in either direction.
One of the equations that a rationalizing \msc[] needs to satisfy is
\begin{align*}
\begin{gathered}
\p{i}{\M}q_{il}+\p{j}{\M}q_{jl}+\p{k}{\M}q_{kl}-\p{l}{\M}(q_{il}+q_{jl}+\frac{2}{3}q_{kl})=0,\\
0.07q_{jl}+0.06q_{kl}=0,
\end{gathered}
\end{align*}
where the second equation results from plugging in the stochastic choice function and simplifying.
Hence, the only solution to the above equation is when $\q{jl}{}=\q{kl}{}=0$.

As in the previous examples, the transition probability matrix should satisfy Equation~\eqref{Qex} and ${\brho[l]{}(I-\mat[]{\M}) = \boldsymbol{0}}$. The following matrix satisfies these requirements for $\bp[\M] = (0.24,0.3,0.22,0.24)$:
\label{app-ex:T3}
\begin{equation*}
Q_{4}(\M)= 
    \begin{pmatrix}
0.4	&0.1	&0.3	&0.2\\
0.1	&0.7	&0.2	&0  \\
0.3 &0.3	&0.4	&0  \\
0.2	&0		&0		&0.8
\end{pmatrix}.
\end{equation*}